\newif\ifANON
\ANONfalse

\documentclass[11pt]{article}
\usepackage[mathlines]{lineno}
\usepackage{paper}

\title{Dimension Reduction for Clustering: \\ The Curious Case of Discrete Centers}

\ifANON
\linenumbers
\author{Anonymous Authors}
\date{}
\else
\author{%
  Shaofeng H.-C. Jiang%
  \thanks{Email: \texttt{shaofeng.jiang@pku.edu.cn}}
  \\ Peking University
  \and
  Robert Krauthgamer%
  \thanks{The Harry Weinrebe Professorial Chair of Computer Science.
    Work partially supported by the Israel Science Foundation grant \#1336/23.
    Email: \texttt{robert.krauthgamer@weizmann.ac.il}
  }
  \\ Weizmann Institute of Science
  \and
  Shay Sapir%
  \thanks{Email: \texttt{shay.sapir@weizmann.ac.il}}
  \\ Weizmann Institute of Science
  \and
  Sandeep Silwal%
  \thanks{Email: \texttt{silwal@cs.wisc.edu}}
  \\ University of Wisconsin-Madison
  \and
  Di Yue%
  \thanks{Email: \texttt{di\_yue@stu.pku.edu.cn}}
  \\ Peking University
}
\fi

\begin{document}

\maketitle

\begin{abstract}
    The Johnson-Lindenstrauss transform is a fundamental method for dimension reduction in Euclidean spaces, that can map any dataset of $n$ points into dimension $O(\log n)$ with low distortion of their distances. This dimension bound is tight in general, but one can bypass it for specific problems. Indeed, tremendous progress has been made for clustering problems, especially in the \emph{continuous} setting where centers can be picked from the ambient space $\mathbb{R}^d$. Most notably, for $k$-median and $k$-means, the dimension bound was improved to $O(\log k)$ [Makarychev, Makarychev and Razenshteyn, STOC 2019].

    We explore dimension reduction for clustering in the \emph{discrete} setting, where centers can only be picked from the dataset, and present two results that are both parameterized by the doubling dimension of the dataset, denoted as $\operatorname{ddim}$. The first result shows that dimension $O_{\epsilon}(\operatorname{ddim} + \log k + \log\log n)$ suffices, and is moreover tight, to guarantee that the cost is preserved within factor $1\pm\epsilon$ for every set of centers. Our second result eliminates the $\log\log n$ term in the dimension through a relaxation of the guarantee (namely, preserving the cost only for all approximately-optimal sets of centers), which maintains its usefulness for downstream applications.

    Overall, we achieve strong dimension reduction in the discrete setting, and find that it differs from the continuous setting not only in the dimension bound, which depends on the doubling dimension, but also in the guarantees beyond preserving the optimal value, such as which clusterings are preserved. 
\end{abstract}

\section{Introduction}
\label{sec:intro}

Oblivious dimension reduction, in the spirit of the Johnson and Lindenstrauss (JL) Lemma~\cite{Johnson1984ExtensionsOL}, 
is a fundamental technique for many Euclidean optimization problems over large, high-dimensional datasets.
It has a strong guarantee: 
there is a random linear map $\pi:\R^d\to \R^t$,
for a suitable target dimension $t=O(\eps^{-2}\log n)$,
such that for every $n$-point dataset $P\subset\R^d$,
with high probability, $\pi$ preserves all pairwise distances in $P$ within factor $1\pm \eps$:
\begin{equation*} 
  \forall x,y\in P, 
  \qquad
  \|\pi(x)-\pi(y)\|\in (1\pm \eps)\|x-y\| ,
\end{equation*}
where throughout $\norm{\cdot}$ is the Euclidean norm.
This guarantee is extremely powerful, particularly for algorithms:
to solve a Euclidean problem on input $P$, 
one can apply the map $\pi$,
solve the same problem on $\pi(P)$, 
which is often more efficient since $\pi(P)$ lies in low dimension, 
and ``lift'' the solution back to the original dimension 
(as discussed further in \cref{sec:comparison_definitions}).

However, many problems require computational resources
that grow exponentially with the dimension (the curse of dimensionality),
and hence even dimension $t=O(\eps^{-2}\log n)$ might be too large. 
Unfortunately, this dimension bound is tight in general,
i.e., for preserving all pairwise distances~\cite{LarsenN17},
but interestingly one may bypass it for \emph{specific} optimization problems,
by showing that the optimal value/solution is preserved 
even when the dimension is reduced beyond the JL Lemma, 
say to dimension $t=O(\eps^{-2})$, which is completely independent of $n$.
This raises an important question:

\begin{tcolorbox}[colback=white,colframe=black,title=]
\begin{center}
For which problems does dimension $o(\eps^{-2}\log n)$ suffice for oblivious dimension reduction?
\end{center}
\end{tcolorbox}
Prior work has revealed an affirmative answer for several key problems, as we discuss below. This paper studies this question for fundamental clustering problems, captured by $(k,z)$-clustering,
which includes the famous $k$-means and $k$-median problems as its special cases. 
In $(k,z)$-clustering, the input is a dataset $P\subset \R^d$, 
and the goal is to find a set of centers $C$ of size $|C|\leq k$ that minimizes 
\[
  \cost^z(P,C)\coloneqq \sum_{p\in P}\dist^z(p,C) ,
  \ \text{ where } \ 
  \dist^z(p,C)\coloneqq \min_{c\in C}\|p-c\|^z.
\]
We can distinguish two variants, differing in their space of potential centers.
In the \emph{continuous} variant, $C$ is a subset of $\R^d$ (the centers lie in the ambient space), 
and in the \emph{discrete} variant, also called sometimes \emph{$k$-medoids},
$C$ is a subset of $P$ (or maybe of a larger set given as input).
A key feature of the discrete version, is that $\pi:P\to \pi(P)$ is invertible, 
hence each potential center in $\pi(P)$ corresponds to a unique potential center in $P$ 
(in contrast, a potential center in the ambient space $\R^t$ has many preimages in $\R^d$). 
Thus, in the discrete version, a set of centers computed for the dataset $\pi(P)$ 
can be mapped back to the higher dimension and serve as centers for the dataset $P$. See \cref{sec:related_works} for a discussion on practical applications of the discrete variant.

The continuous variant is a success story of the ``beyond JL'' program.
A series of papers \cite{DBLP:conf/nips/BoutsidisZD10, DBLP:conf/stoc/CohenEMMP15, DBLP:conf/stoc/BecchettiBC0S19, MakarychevMR19}
has culminated showing that 
target dimension $t=O(\eps^{-2}\log\tfrac{k}{\eps})$,
which is independent of $n$, 
suffices to preserve all the solutions within factor $1\pm\eps$.
Curiously, Charikar and Waingarten~\cite{CW22} observed that 
the discrete variant \emph{behaves very differently}: 
certain instances require $t=\Omega(\log n)$, even for $k=1$
(when using the standard Gaussian-based map $\pi$). 
Counterintuitively, restricting the centers to be data points
makes dimension reduction significantly harder!

To bypass this limitation, we consider the doubling dimension,
which was identified in previous work as a natural parameter 
that is very effective in achieving ``beyond JL'' bounds~\cite{IndykN07,NarayananSIZ21,JKS23,HuangJKY25,GJKSSSW25_maxmatching}.
Formally, the \emph{doubling dimension} of $P$, denoted $\ddim(P)$, 
is the smallest positive number such that every ball in the finite metric $P$
can be covered by $2^{\ddim(P)}$ balls of half the radius.
For several problems, including nearest neighbor~\cite{IndykN07}, 
facility location~\cite{NarayananSIZ21,HuangJKY25}, 
and maximum matching~\cite{GJKSSSW25_maxmatching}, 
target dimension $t=O(\eps^{-2}\log\tfrac{1}{\eps}\cdot \ddim(P))$ suffices.
Note that restricting the doubling dimension does not immediately imply
a better dimension reduction of the JL flavor, 
as there are datasets $P\subset \R^d$ with $\ddim(P)=O(1)$ where no linear map 
can approximately preserve all pairwise distances (see e.g., \cite[Remark 4.1]{IndykN07}).

\subsection{Main results}\label{sec:main_results}
We present the first dimension reduction results for discrete $(k,z)$-clustering, along with matching lower bounds.
Our first result (\Cref{thm:informal_main_1}) provides a strong approximation guarantee, %
but requires a $\log\log n$ term in the target dimension, which we show is necessary. 
Our main result (\Cref{thm:informal_main_2}) avoids this $\log\log n$ term, 
through a relaxation of the guarantee that maintains its algorithmic usefulness,
e.g., it still implies that the optimal value is preserved up to factor $1\pm\eps$.

In all our results, the random linear map $\pi$ is given by 
a matrix $G\in\R^{t\times d}$ of iid Gaussians $N(0,\tfrac{1}{t})$, 
which we refer to as a \emph{Gaussian JL map}.
This is nowadays a standard JL map~\cite{IndykM98,DasguptaG03}, 
and our results may extend to other JL maps, similarly to prior work in this context.
We denote the optimal value of discrete $(k,z)$-clustering by
\[
  \opt^z(P) = \min_{C\subset P, |C|=k} \cost^z(P,C) , 
\]
however for sake of exposition, we omit $z$ and focus on $z=1$ or $z=2$, 
which are discrete $k$-median and $k$-means.
We use the notation $\tilde{O}(f)$ to hide factors that are logarithmic in $f$, 
although below it only hides a $\log\tfrac{1}{\eps}$ factor.

\begin{theorem}[Informal version of \Cref{thm:k_medoid_forall_centers_partitions}]\label{thm:informal_main_1}
    For suitable
    $t=\tilde{O}(\eps^{-2}(\ddim(P)+ \log k + \log\log n))$,
    with probability at least $2/3$,
    \begin{enumerate}
        \item \label{item:opt_ub} $\opt(G(P))\leq (1+\eps)\opt(P)$, and
        \item \label{item:opt_lb} for all $C\subseteq P, |C|\leq k$, we have  $\cost(G(P),G(C))\geq (1-\eps)\cost(P,C)$.
    \end{enumerate}
\end{theorem}
This theorem has immediate algorithmic applications.
First, it implies that the optimal value is preserved, 
i.e., $\opt(G(P))\in (1\pm\eps)\opt(P)$. 
Second, for every $C\subset P$ and $\beta>1$,
if the set of centers $G(C)$ is a $\beta$-approximate solution for the instance $G(P)$, 
then $C$ is a $(1+O(\eps))\beta$-approximate solution for the instance $P$.
Therefore, the theorem fit into the general paradigm of using oblivious linear maps --- apply the mapping, solve the problem in low dimension, and lift the centers back to the higher dimension.

It is interesting to compare our result with the continuous variant of $(k,z)$-clustering.
On the one hand, to preserve the optimal value in the continuous variant,
we know from~\cite{MakarychevMR19} that
target dimension $O(\eps^{-2}\log\tfrac{k}{\eps})$ suffices, independently of $\ddim(P)$.
On the other hand, \Cref{thm:informal_main_1} further provides a ``for all centers'' guarantee, 
which is not attainable in the continuous version (by any linear map),
by simply considering centers in the kernel of the linear map (see \Cref{thm:LB_continuos_forall_centers}).
We examine and discuss these guarantees more carefully in \cref{sec:comparison_definitions}.

\paragraph{Matching lower bounds.}
The results in \Cref{thm:informal_main_1} are nearly tight for Gaussian JL maps, and likely for all oblivious linear maps. 
It is known that achieving $\opt(G(P))\in (1\pm\eps)\opt(P)$
requires target dimension $t=\Omega(\log k)$,
even for a dataset $P$ of doubling dimension $O(1)$ \cite{NarayananSIZ21},
and another known lower bound is that $t=\Omega(\ddim(P))$, even for $k=O(1)$ \cite{CW22}.
It is easy to tighten these bounds with respect to the dependence on $\eps$,
and we include it in \Cref{sec:tighten_known_lowerbound} for completeness.
We complete the picture, and show in \Cref{thm:lowerbound_loglogn_medoid_forall_c}
the multiplicative approximation of \Cref{thm:informal_main_1}
requires dimension $t=\Omega(\eps^{-2}\log\log n)$, 
even for $k=1$ and a dataset $P$ of doubling dimension $O(1)$.

To get some intuition about the discrete variant, 
we briefly recall the hard instance of \cite{CW22},
taking $z=1$ for simplicity. 
Consider $k=2$, and let $P$ be the first $n$ standard basis vectors,
thus $\ddim(P)=\log n$.
The pairwise distances all equal $\sqrt{2}$, hence $\opt(P)=\sqrt{2}\cdot n$. 
The standard basis vectors form a well-known hard instance for the JL Lemma, hence,
when using target dimension $t=o(\eps^{-2}\log n)$, with high probability, there exists $j_1\in [\tfrac{n}{2}]$ such that $\|Ge_{j_1}\|<1-10\eps$.  Similarly, let $j_2>\tfrac{n}{2}$ be such an index for the last $\tfrac{n}{2}$ standard basis vectors.
Let $Ge_{j_1},Ge_{j_2}$ be the two centers for $G(P)$, 
and assign the first $\tfrac{n}{2}$ basis vectors to $Ge_{j_2}$
and the last $\tfrac{n}{2}$ vectors to $Ge_{j_1}$. 
Now a simple argument using the independence between the two halves,
see \Cref{sec:tighten_known_lowerbound}, shows that 
$\opt(G(P))\leq (1-\eps)\sqrt{2} \cdot n=(1-\eps)\opt(P)$ with probability $2/3$.

\paragraph{A relaxed guarantee.}
Our main result avoids the $\log\log n$ term in \Cref{thm:informal_main_1}
by slightly relaxing the guarantee, while keeping it useful for downstream applications.

\begin{theorem}[Informal version of \Cref{thm:k-medoid_no_loglogn_improved}]\label{thm:informal_main_2}
    For suitable 
    $t=\tilde{O}(\eps^{-2}(\ddim(P)+ \log k))$,
    with probability at least $2/3$,
    \begin{enumerate}
        \item $\opt(G(P))\leq (1+\eps)\opt(P)$, and
        \item for all $C\subseteq P, |C|\leq k$, we have  $\cost(G(P),G(C))\geq \min\{(1-\eps)\cost(P,C),100\opt(P)\}$.
    \end{enumerate}
\end{theorem}

This theorem implies that the optimal value is preserved, 
i.e., $\opt(G(P))\in (1 \pm\eps)\opt(P)$. 
Let us further examine which solutions are preserved under this guarantee:
For all $C\subset P$ and $1<\beta< \tfrac{100}{1+\eps}$, 
if the set of centers $G(C)$ is a $\beta$-approximate solution for the instance $G(P)$, 
then $C$ is a $(1+O(\eps))\beta$-approximate solution for the instance $P$.
Recall that for \Cref{thm:informal_main_1}, we had a similar claim, but without the restriction $\beta< \tfrac{100}{1+\eps}$.
The constant $100$ here is arbitrary, and can be changed to any $\alpha>2$, 
at the cost of increasing the target dimension by an additive $O(\eps^{-2}\log\log \alpha)$ term.

\subsection{Various notions for preserving solutions}
\label{sec:comparison_definitions}

We study several definitions for dimension reduction for $k$-clustering. 
All these definitions require (perhaps implicitly) that $\opt(G(P))\leq (1+\eps)\opt(P)$, 
i.e., that the optimal value has bounded expansion.
This direction is often easy because it suffices to analyze one optimal solution for $P$.
In the other direction, one may naively require that $\opt(G(P))\geq (1-\eps)\opt(P)$,
however this is rather weak, as it does not guarantee that \emph{solutions} are preserved.
Moreover, even requiring that an optimal solution for $G(P)$ is a near-optimal solution for $P$ is quite limited,
because a near-optimal solution for $G(P)$,
say one found by a $(1+\eps)$-approximation algorithm, 
may be lifted to a poor solution for $P$. 
In fact, such a phenomenon was observed for minimum spanning tree (MST) when using target dimension $t=o(\log n)$:
an optimal MST of $G(P)$ is a $(1+\eps)$-approximate MST of $P$,
however a $(1+\eps)$-approximate MST of $G(P)$ may have large cost for $P$ \cite{NarayananSIZ21}.
Ideally, we want the cost of \emph{every solution} to have bounded contraction, 
as it allows to lift any solution for $G(P)$ to a solution for $P$,
and we thus consider several different notions for the set of solutions, as follows.
For simplicity, we present these for $z=1$ in the discrete setting, 
but they extend naturally to all $z\geq 1$ and to the continuous setting.
\begin{enumerate}%
    \item \textbf{Partitions.} 
    A solution is a partition $\P=(P_1,\ldots,P_k)$ of $P$. Its cost is defined as
    $\cost(\P) \coloneqq \sum_{i=1}^k\min_{c\in P_i}\sum_{p\in P_i}\|p-c\|$.
    \item \textbf{Centers.} 
    A solution is a set of centers $C=(c_1,\ldots,c_k)\subseteq P$. Its cost is defined as
    $\cost(P,C) \coloneqq \sum_{p\in P}\dist(p,C)$.
    \item \textbf{Centers and partitions.} 
    A solution is a partition $\P=(P_1,\ldots,P_k)$ of $P$ and a set of centers $C=(c_1,\ldots,c_k)\subseteq P$. Its cost is defined as
    $\cost(\P,C) \coloneqq \sum_{i=1}^k\sum_{p\in P_i}\|p-c_i\|$.
\end{enumerate}
These definitions are fairly natural, and were used in prior work on dimension reduction,
e.g., partition-based solutions were used in \cite{MakarychevMR19} for $k$-means and $k$-median, 
and center-based solutions were used in \cite{JKS23} for $k$-center.
It was observed in \cite{CW22} that not all ``for all'' guarantees are the same; 
in particular, ``for all centers'' and ``for all partitions'' are incomparable.
However,``for all centers and partitions'' is clearly stronger than both.

Next, we define contraction for solutions, capturing the two notions in \Cref{thm:informal_main_1,thm:informal_main_2}.
The notion in \Cref{thm:informal_main_1} is simply of \emph{multiplicative contraction}: A solution $S$ has $(1-\eps)$-contraction if $\cost(G(S))\geq (1-\eps)\cost(S)$.
The notion in \Cref{thm:informal_main_2} is new, at least in the context of dimension reduction, and goes as follows. %
\begin{definition}[Relaxed Contraction]
    A solution $S$ has $\alpha$-relaxed $(1-\eps)$-contraction 
    (for $\alpha>1$, $\eps>0$)
    if $\cost(G(S))\geq \min\{\alpha\opt(P),(1-\eps)\cost(S)\}$.
\end{definition}

Using these definitions, we can restate \Cref{thm:informal_main_1} as having $(1-\eps)$-contraction for all centers, 
and restate \Cref{thm:informal_main_2} as achieving $100$-relaxed $(1-\eps)$-contraction for all centers.
In fact, we can strengthen \Cref{thm:informal_main_1} 
to assert $(1-\eps)$-contraction for all \emph{centers and partitions}.
\begin{theorem}[Strengthened \Cref{thm:informal_main_1}, informal]\label{thm:informal_partitions_and_centers}
    For suitable
    $t=\tilde{O}(\eps^{-2}(\ddim(P)+ \log k + \log\log n))$,
    with probability at least $2/3$, for all partitions $\P=(P_1,\ldots,P_k)$ of $P$ and sets of centers $C=(c_1,\ldots,c_k)\subseteq P$,
    \[
    \cost(G(\P),G(C))\geq (1-\eps)\cost(\P,C).
    \]
\end{theorem}
This strengthening is not attainable for \Cref{thm:informal_main_2}, 
as dimension $\Omega(\eps^{-2}\log\log n)$ is needed to get a ``for all centers and partitions'' guarantee, even for relaxed contraction 
(see \Cref{thm:lowerbound_loglogn_medoid_forall_partitions_and_c}).
However, we do not know if a ``for all partitions'' guarantee is possible without the $\log\log n$ term. If it is possible, then a curious phenomenon will occur: we get a ``for all partitions'' and a ``for all centers'' guarantees, but not a ``for all centers and partitions'' guarantee.
All our results are summarized in \Cref{table:results}.

\paragraph{Candidate centers.}
We consider also a more general variant of $k$-clustering, 
where the candidate centers are part of the input (given either explicitly or implicitly):
Given a dataset $P$ and candidate-centers set $Q$, 
the goal is to find $C\subseteq Q$ of size $|C|\leq k$ 
that minimizes $\sum_{p\in P} \dist^z(p,C)$.
When $Q=\R^d$ or $Q=P$, we obtain the continuous and discrete variants, respectively.

We observe a slightly different phenomenon in terms of the attainable contraction: to get $(1-\eps)$-contraction, 
one needs target dimension $\Theta(\eps^{-2}\log |Q|)$, 
and the lower bound holds even when both $P$ and $Q$ are doubling and $k=1$.
We can still obtain claims analogous to \Cref{thm:informal_main_2,thm:informal_partitions_and_centers}, albeit with relaxed contraction:
a ``for all partitions and centers'' using dimension $t=\tilde{O}(\eps^{-2}(\ddim(P\cup Q)+\log k + \log\log n))$, 
and a ``for all centers'' for the same target dimension but without the $\log\log n$ term.
See \Cref{table:results} for references.

\begin{table}[t]
\begin{tabular}{|l|l|l|l|l|l|}
\hline
Problem    & Target dimension                                              & $\forall$ partitions & $\forall$ centers & contraction & Reference \\
\hline
Continuous & $O(\eps^{-2}\log k)$                                          & yes                  & no                 & multiplicative                   & \cite{MakarychevMR19}       \\
           & $\Omega(\eps^{-2}\log k)$                                     & no                   & no                 &          even for value          &    \cite{NarayananSIZ21}       \\
           & $>d-1$                                                      & no                   & yes                 & even for relaxed                  & Thm \ref{thm:LB_continuos_forall_centers}      \\
\hline
Discrete   & $O(\eps^{-2}(\ddim+\log k +\log\log n))$ & yes                  & yes               & multiplicative                    &     Thm \ref{thm:k_medoid_forall_centers_partitions}     \\
           & $O(\eps^{-2}(\ddim+\log k))$             & no                   & yes                & relaxed                  &    Thm \ref{thm:k-medoid_no_loglogn_improved}       \\
           & ?         & yes                  & no                 & any                   & \underline{OPEN}     \\
           & $\Omega( \eps^{-2}\log\log n)$                & yes                  & yes                & even for relaxed                  &     Thm \ref{thm:lowerbound_loglogn_medoid_forall_partitions_and_c}      \\
           & $\Omega(\eps^{-2}\log\log n)$                  & no                   & yes               & multiplicative                    &     Thm \ref{thm:lowerbound_loglogn_medoid_forall_c}     \\
           & $\Omega(\eps^{-2}\log k)$ & no & no & even for value & \cite{NarayananSIZ21} \\
           & $\Omega(\eps^{-2}\ddim)$ & no & no & even for value & \cite{CW22} \\
\hline
Candidate & $O(\epsilon^{-2} \log s)$ & yes & yes & multiplicative & Thm \ref{thm:k_medoids_for_all_partitions_c_candidate_multi} \\
  
 centers &      $O(\epsilon^{-2}(\ddim + \log k + \log \log n))$ & yes & yes & relaxed & Thm \ref{thm:k_medoids_for_all_partitions_c_candidate} \\
 & $O(\epsilon^{-2} (\ddim + \log k))$ & no & yes & relaxed & Thm \ref{thm:k-medoid_no_loglogn_improved} \\
 
 &                    $\Omega(\epsilon^{-2} \log s)$ & no & yes & multiplicative & Thm \ref{thm:lowerbound_logn_medoid_forall_c_candidate} \\
\hline
\end{tabular}
\caption{Summary of our results for dimension reduction for $k$-clustering. 
The notions of ``for all'' centers and/or partitions, and of multiplicative/relaxed contraction are as explained in \Cref{sec:comparison_definitions}. Some lower bounds apply even for preserving the optimal value; for clarity, it is noted in the table they hold ``even for value''.
In the setting of candidate centers, the size of the candidate set is denoted by $s$. Suppressing $\log\tfrac{1}{\eps}$ terms and the dependence on $\alpha$ for $\alpha$-relaxed contraction.}\label{table:results}
\end{table}

\subsection{Other related work}
\label{sec:related_works}

Besides the aforementioned results for ``beyond JL'' for clustering problems, 
there are also several improved bounds for other classes of problems such as Max-Cut \cite{DBLP:conf/wads/LammersenSS09, lammersen2010approximation,DBLP:conf/stoc/ChenJK23}, numerical linear algebra \cite{mahoney2011randomized, Woodruff14, CohenNW16}, 
and other applications \cite{bartal2011dimensionality, GK15, izzo2021dimensionality}. 

The discrete $k$-median problem in Euclidean space was originally shown to be NP-hard by Papadimitriou, even for the case of $d = 2$ \cite{papadimitriou1981worst}. In terms of hardness of approximation, the current state of the art is that one cannot approximate the discrete $k$-means or $k$-median problem beyond $1.07$ and $1.17$, respectively, assuming $\text{P} \ne \text{NP}$ \cite{cohen2019inapproximability, cohen2022johnson}. 
As for upper bounds, the best approximation factors known in polynomial time are $2+\eps$ for any fixed $\eps > 0$ for discrete Euclidean $k$-median \cite{cohen20252+} and $ 5.912$ for discrete Euclidean $k$-means \cite{cohen2022improved}. 
There are also algorithms that achieve $1+\eps$ approximation (again in the discrete case) in time that is doubly exponential in the doubling dimension, see \cite{cohen2021near} for a thorough discussion.

The discrete variant that we study may also be preferred over the continuous version in certain applications.  First, it is thought to be less sensitive to outliers in practice than the continuous version \cite{ParkJ09, kaufman2009finding}. 
Second, in applications where cluster centers are used as data summarization, 
interpretability might require the centers to be part of the dataset. 
For example, in applications based on machine-learning embeddings of objects such as text \cite{xie2016unsupervised}, an arbitrary vector in the embedding space might not represent any actual object. 
A similar issue arises for structured data such as sparse data or images, 
e.g., the ``average image'' is visually random noise \cite{leskovec2020mining, tiwari2020banditpam} or the average of sparse vectors is not necessarily sparse. 
A discrete center, however, represents an actual underlying object, and thus preserves the underlying properties of the input points.

\subsection{Technical overview}
Since the dimension-error tradeoff behaves differently between the discrete and continuous settings,
it is not surprising that our results for the discrete setting require new techniques.
To simplify the discussion,
we focus on the $k$-medoids ($z = 1$) case,
and an alternative guarantee that only preserves the optimal \emph{value}, i.e., 
\begin{equation}\label{eqn:preserve_value}
    \opt(G(P)) \in (1 \pm \epsilon) \opt(P),
\end{equation}
with target dimension bound $t=\tilde{O}(\eps^{-2}(\ddim(P)+ \log k))$ which is the same to that in \Cref{thm:informal_main_2}.
While this is a weaker guarantee than both \Cref{thm:informal_main_1} and \Cref{thm:informal_main_2},
it already introduces major technical challenges,
and the techniques for this claim covers most of our new ideas.

We begin our discussion with the case $k = 1$. 
We first argue that even for this case, a natural framework based on extension theorems
(which has been used in previous works on dimension reduction for clustering) fails in our discrete case.

\paragraph{Failure of extension theorems in the discrete setting.}
To prove \eqref{eqn:preserve_value} (and possibly more general claims),
a natural framework based on \emph{extension theorems} have been widely used in dimension reduction for clustering.
Specifically, given an arbitrary center $v$ in the target space (e.g., $v$ is the optimal $1$-median center of $G(P)$),
one can define an ``inverse image'' $u$ in the original space such that $\cost(P, u) \leq (1 + \epsilon) \cost(G(P), v)$,
and this directly implies $\opt(G(P)) \geq \tfrac{1}{1 + \epsilon} \opt(P)$.
The key step of defining ``inverse image'' is precisely what an extension theorem does.
This framework is widely used in prior works such as~\cite{MakarychevMR19,JKS23},
in the spirit of the classic Kirszbraun extension theorem~\cite{Kirszbraun1934} 
or the robust one-point extension theorem~\cite[Theorem 5.2]{MakarychevMR19}. 
However, such extension theorems are only known to work in the continuous setting,
which require to pick the inverse image $v \in \mathbb{R}^d$ from the entire $\mathbb{R}^d$
and cannot be restricted only to the data points $v \in P$.%
\footnote{
    We note that the Kirszbraun theorem may be adapted to work for the discrete case when the target dimension $t = O(\log n)$, but this dimension bound is too large to be useful.
}

\paragraph{Our techniques.}
We start with $k = 1$ case (a detailed discussion can be found in \Cref{sec:k_eq_1}).
In this case, we first obtain a target dimension bound with an $O(\log \log n)$ factor,
by utilizing the existence of a small movement-based coreset.
A coreset is a small accurate proxy of the dataset, 
and the movement-based coreset additionally requires the existence of a ``local'' mapping such that each data point can be mapped to a nearby coreset point.
The dimension reduction simply preserves the pairwise distance on the coreset,
and \eqref{eqn:preserve_value} is argued via the local mapping.
A conceptually similar coreset-to-dimension-reduction idea has also been employed in~\cite{CW22},
and one main difference is that we also utilize the movement/local property of the coreset.

Then, to remove the $O(\log \log n)$ factor,
we consider a weaker guarantee as in \Cref{thm:informal_main_2},
where we prove the $(1 + \epsilon)$ relative error only for near-optimal solutions,
and for the other solutions we have a flat $100 \opt(P)$ error.
This relaxed guarantee is strong enough for \eqref{eqn:preserve_value} (and many other applications),
which may be of independent interest to further studies.
Our analysis is crucially built on this small vs large cost case,
albeit we also need to consider the middle ground of the mix of the two.

Finally, we discuss the generalization to $k > 1$ in \Cref{sec:k_ge_2},
which introduces several nontrivial technical complications from $k = 1$.

\subsubsection{The $k$ = 1 case}
\label{sec:k_eq_1}

The easy side of \eqref{eqn:preserve_value} 
is the upper bound $\opt(G(P))\leq (1+\eps)\opt(P)$,
even for the general $k$ case.
The reason is that it suffices to preserve the cost w.r.t. an optimal center set $C^*$,
and since $C^*$ is a fixed solution, even a target dimension $t = O(\epsilon^{-2} \log(1/\epsilon))$ will be sufficient.
This is a standard argument also observed in prior works.
The lower bound $\opt(G(P)) \geq (1 - \epsilon) \opt(P)$ is the major challenge.
To prove this inequality, we want to preserve the clustering cost w.r.t. the optimal center set of $G(P)$, denoted by $C$.
Since $C$ is a random set that depends on $G$, preserving its cost is almost the same as preserving the cost of \emph{all} center sets, which is exactly the guarantee \ref{item:opt_lb} of \Cref{thm:informal_main_1,thm:informal_main_2}.

To introduce our new techniques, we first establish a weaker target dimension bound of $O(\epsilon^{-2}(\ddim + \log\log n))$, and this part contains main ideas for proving \Cref{thm:informal_main_1}. 
We then overview the key steps to eliminate the extra $\log \log n$ term, which also reflects how we prove \Cref{thm:informal_main_2}.

\paragraph{The $O(\log\log n)$ bound: from coreset to dimension reduction.}
To prove \eqref{eqn:preserve_value}, we use an approach inspired by the movement-based coreset construction in Euclidean spaces~\cite{Har-PeledM04}.
Roughly speaking, a movement-based coreset%
\footnote{This definition is tailored to our need and may be slightly different to that in the literature.}
is a subset $S \subseteq P$, such that there exists a mapping $\sigma \colon P \to S$ satisfying $\sum_{p \in P} \|p - \sigma(p)\| \leq O(\epsilon) \opt(P)$.
Our framework is summarized as follows: we first construct a movement-based coreset $S$ to compress the dataset $P$. 
Next, we apply the standard JL lemma to preserve pairwise distances in the coreset $S$ within $(1 \pm \epsilon)$, which requires $O(\epsilon^{-2} \log |S|)$ target dimensions.
After this step, the optimal value of $S$ is already preserved, nemely, $\opt(G(S)) \in (1 \pm \epsilon) \opt(S)$.
Finally, it suffices to show 
that the cost of snapping data points to their nearest neighbor in $S$ (i.e., $\sum_{p \in P} \|p - S(p)\|$ and $\sum_{p \in P} \|Gp - GS(p)\|$) is negligible in both original and target spaces.

The construction of the coreset is essentially the same with that in~\cite{Har-PeledM04},
except that~\cite{Har-PeledM04} also assigns weight to the coreset points and here we only need the point set itself.
We review the construction. This construction is based on a sequence of \emph{nets}, a standard tool for discretizing metrics.
Formally, a $\rho$-net of a point set $P$ is a subset $N \subseteq P$, such that 1) the interpoint distances in $N$ are at least $\rho$, and 2) every point in $P$ has a point in $N$ within distance $\rho$. 
(See the more detailed definition in \cref{def:net}).
Denote $c^* \in P$ as an optimal discrete $1$-median center.
We construct nets on a sequence of balls centered at $c^*$ with geometrically decreasing radii.
Denote $r_0 := \opt(P)$ and $r_\ell := r_0/2^\ell$ for $\ell = 1, 2, \dots, \log n$.
Construct the level $\ell$ net $N_\ell$ as an $\epsilon r_\ell$-net on the ball $B(c^*, r_\ell)$, and
denote $N := \bigcup_{\ell = 0}^{\log n} N_\ell$ to be the union of all $\log n$ levels of nets.

By the standard packing property of doubling metrics, each net has size $|N_\ell| \leq O(\epsilon^{-O(\ddim)})$, thus $|N| \leq O(\epsilon^{-O(\ddim)} \log n)$,
which implies a target dimension $t = O(\epsilon^{-2}(\ddim \log \epsilon^{-1} + \log \log n))$.
On the other hand, let $G(c) \in G(P)$ be an optimal discrete $1$-median center of $G(P)$.
Then the total cost of snapping $c$ and all data points to the nearest neighbor in $N$ 
(i.e., $\sum_{p \in P} (\|p - N(p)\| + \|c - N(c)\|)$) can be bounded by $O(\epsilon) (\opt(P) + \cost(P, c))$ in the original space.
Based on results in~\cite{IndykN07}, we further show that this snapping cost in the target space (i.e., $\sum_{p \in P} (\|Gp - GN(p)\| + \|Gc - GN(c)\|)$) can increase by at most a constant factor.

Finally, we note that the above analysis can be applied to obtain the ``for all centers'' guarantee in \Cref{thm:informal_main_1}, or even the stronger ``for all centers and partitions'' guarantee in \Cref{thm:informal_partitions_and_centers}.

\paragraph{Removing the $\log \log n$ term via relaxed guarantee.}
Let us first recall the cause of the $\log \log n$ term.
We  the JL Lemma to $N$, which is a union of $\log n$ nets, each of size $\epsilon^{-O(\ddim)}$.
The $\log \log n$ thus comes from a union bound over all $\log n$ levels.
To bypass this union bound, we use two technical ideas.
First, we avoid touching cross-level pairs and only apply the union bound for each $N_\ell$ separately.
This requires us to always snap $p$ and $c$ to the same level of net when handling each $p \in P$.
Second, for a single level, we analyze its maximum distance distortion which is a random variable, and bound the expectation.
We remark that some levels will be distorted significantly, but the average distortion is $(1+O(\eps))$.
Similar ideas have been used by prior works (e.g.,~\cite{GJKSSSW25_maxmatching}).

Consider the following two extremes.
First, suppose $c$ is the closest point to $c^*$, say, $\forall p \in P, \|c - c^*\| \leq \|p - c^*\|$.
For every $p \in P$, we can snap $p$ to its nearest neighbor in net $N_p$.
Observe that $c$ can also be covered by $N_p$.
The cost of snapping $p$ and $c$ can both be bounded by $O(\epsilon) \cdot \|p - c^*\|$, and we show that \emph{on average}, the cost of snapping $Gp$ and $Gc$ is bounded by $O(\epsilon) \cdot \|p - c^*\|$ as well, which adds up to $O(\epsilon) \opt(P)$.
The other extreme is that $c$ is very far from $c^*$, i.e, $\|c - c^*\| > \opt(P)/10$.
In this case, we can no longer snap $c$ to the same net as $p$ (like the previous case).
We show that in this case, $\cost(G(P), Gc)\geq 100\opt(P)$.

If $c$ does not fall into any of the above two extremes, our analysis is a combination of them.
Indeed, we show the \emph{relaxed} ``for all centers'' guarantee,
\begin{equation}\label{eqn:relaxed_forall_1center}
    \forall c \in P, \qquad \cost(G(P), Gc) \geq \min\{(1 - \epsilon) \cost(P, c), 100 \opt(P)\}.
\end{equation}
Note that this is exactly the same as the guarantee \ref{item:opt_lb} of \Cref{thm:informal_main_2}, and that the two terms in the $\min$ corresponds to the aforementioned two extremes, respectively.
Specifically, we first specify a level $\ell$ and its corresponding radius $r_\ell$.
If $\|c - c^*\| > r_\ell$, then we fall into the second extreme and show that $\cost(G(P), Gc) \geq 100 \opt(P)$.
Otherwise, $\|c - c^*\| \leq r_\ell$, then we handle each $p \in P$ differently, depending on the distance $\|p - c^*\|$.
If $\|p - c^*\| \geq r_\ell$, then we use the same argument as the first extreme --- snapping both $p$ and $c$ to $N_p$, bounding the snapping cost, and analyzing the additive contraction.
If $\|p - c^*\| < r_\ell$, then we snap both $p$ and $c$ to $N_\ell$.
Since $\ell$ is a fixed level, a union bound over $N_\ell$ is affordable and we obtain $\cost(G(P), Gc) \geq (1 - \epsilon) \cost(P, c)$ in this case.

\subsubsection{Generalization to $k > 1$}
\label{sec:k_ge_2}

Instead of directly generalizing \eqref{eqn:relaxed_forall_1center}, 
we first show a weaker guarantee: for target dimension $t = O(\epsilon^{-2} \ddim \log k)$, 
\begin{equation}\label{eqn:wrong_assignment}
    \forall C \subseteq P, |C| = k, \qquad \sum_{p \in P} \|Gp - GC(p)\| \geq \min\{(1 - \epsilon) \cost(P, C), 100 \opt(P)\},
\end{equation}
where $C(p)$ is the center in $C$ closest to $p$.
Note that \eqref{eqn:wrong_assignment} is weaker than what we desire in \Cref{thm:informal_main_2}, for the following two reasons.
First, the target dimension is worse than the $O(\epsilon^{-2} (\ddim + \log k))$ in \Cref{thm:informal_main_2}.
Second, the left hand side of \eqref{eqn:wrong_assignment} can be much larger than $\cost(G(P), G(C))$, since the image of $C(p)$ under $G$ (i.e., $GC(p)$) is not necessarily the nearest neighbor of $Gp$ in $G(C)$.
Nonetheless, the proof of \eqref{eqn:wrong_assignment} already captures most of our key ideas.
In the end of this section, we briefly discuss how we obtain a sharper target dimension bound as well as a stronger guarantee.

Suppose $C^* \subseteq P$ is an optimal solution, which induces a clustering $\mathcal{C}^* = \{S_1^*, S_2^*, \dots, S_k^*\}$.
Our general proof framework is the same as the $k = 1$ case --- considering the ``distance'' between $C$ and $C^*$, if $C$ is ``far from'' $C^*$, then we show $\cost(G(P), G(C)) \geq 100 \opt(P)$; otherwise we show $\cost(G(P), G(C)) \geq (1 - \epsilon) \cost(P, C)$.

However, an immediate issue is how to define that $C$ and $C^*$ are far from or close to each other.
For each $i \in [k]$, we specify a ``threshold level'' of cluster $S_i^*$, denoted by $\ell_i$. 
We say $C$ is ``far from'' $C^*$ if there exists $i \in [k]$, such that $\dist(c_i^*, C) > 10 r_{\ell_i}$.
In this case, the cost of connecting $B(c_i^*, r_{\ell_i})$ to $C$ is already high.
We further prove that $\cost(G(P), G(C)) \geq 100 \opt(P)$, by careful analysis of the randomness of $G$.

Now suppose $C$ is ``close to'' $C^*$, i.e., $\forall i \in [k], \dist(c_i^*, C) \leq 10 r_{\ell_i}$.
Our key observation is that for every $p \in S_i^*$, $C(p)$ should also be close to $c_i^*$, i.e., 
\begin{equation}\label{eqn:key_observation}
    \forall p \in S_i^*, \qquad 
    \|C(p) - c_i^*\| \leq O(\max\{\|p - c_i^*\|, r_{\ell_i}\}).
\end{equation}
As a natural generalization of the $k = 1$ case, we lower bound $\|Gp - GC(p)\|$ for $p \in S_i^*$ differently, depending on the distances $\|C(p) - c_i^*\|$.
If $\|C(p) - c_i^*\| \geq r_{\ell_i}$, then we snap both $p$ and $C(p)$ to the (enlarged) net $N_p$.
(We can do this since \eqref{eqn:key_observation} holds.)
Otherwise, 
we snap both $p$ and $C(p)$ to the (enlarged) net $N_{\ell_i}$.
The snapping cost and the distance contraction are bounded similarly to the $k = 1$ case.
This simply introduces an extra $\log k$ factor in the target dimension.

\paragraph{Decoupling $\ddim$ from $\log k$.}
So far, we only obtain an $O_\epsilon(\ddim \log k)$ bound, instead of $O_\epsilon(\ddim + \log k)$.
This is due to error accumulation:
Recall we handle each (optimal) cluster $S_i^*$ separately, each of which incurs an $O(\epsilon) \opt(P)$ additive error;
hence, we have to rescale $\epsilon$ by a $1/k$ factor to compensate the accumulated error of $k$ clusters, 
resulting in an $O(\epsilon^{-2} \ddim \log k)$ target dimension (na\"ively, that results in $\tilde{O}(\epsilon^{-2} k^2\ddim)$ target dimension, but this is avoided by an easy adaptation).

To decouple these two factors, we need more delicate analysis for the error.
For ``far'' points $p \in S_i^*$ with $\|C(p) - c_i^*\| \geq r_{\ell_i}$, the snapping and distortion error is $O(\epsilon)  \|p - c_i^*\|$ in expectation, which adds up to $O(\epsilon) \opt(P)$ and does not incur any error accumulation.
However, the error accumulation happens for ``close'' points $p$ with $\|C(p) - c_i^*\| < r_{\ell_i}$, where the snapping cost within a single cluster $S_i^*$, 
namely $\sum_{p \in S_i^*} \|p - N_\ell(p)\|$,
is already $O(\epsilon) \opt(P)$, which accumulates to $O(k \epsilon) \opt(P)$.

To reduce the error accumulation, we further divide the close points (i.e., $\|C(p) - c_i^*\| < r_{\ell_i}$) into two ranges,
namely, the \emph{close range} $\|C(p) - c_i^*\| < r_{\ell_i}/k$ 
and the \emph{middle range} $\|C(p) - c_i^*\| \in [r_{\ell_i}/k, r_{\ell_i}]$,
and handle these two ranges differently.
The cost of points in the close range can be bounded by $O(\epsilon/k) \opt(P)$, which adds up to $O(\epsilon) \opt(P)$.
For points in the middle range, we handle them in a point-by-point manner, at the cost of $\poly(k) e^{-\Omega(\epsilon^2 t)}$ per point.
Since there are at most $k \cdot O(\log k)$ levels in the middle range, a union bound over all net points at these levels will be affordable.

\paragraph{Handling nearest neighbor assignment in the target space.}
Recall that \eqref{eqn:key_observation} conerns the cost $\|Gp - GC(p)\|$,
which is the cost in the target space with respect to the nearest neighbor assignment in the \emph{original} space.
However, what we really need is the nearest neighbor assignment in the \emph{target} space.
To capture such misalignment in the original and target spaces, we define a mapping $f$ to be the assignment in the target space, i.e., $f(p)$ is the center in $C$ realizing $\dist(Gp, G(C))$, so that $\cost(G(P), G(C)) = \sum_{p \in P} \|Gp - Gf(p)\|$, and $f(p) = C(p)$ does not hold in general.
We attempt to modify the previous analysis to lower bound each $\|Gp - Gf(p)\|$ instead of $\|Gp - GC(p)\|$.

To lower bound this distance, we attempt to replace every $C(p)$ with $f(p)$ in our previous proof.
The analysis becomes problematic, as our structural observation \eqref{eqn:key_observation} no longer holds if we change $C(p)$ to $f(p)$, and this turns out to be the only place where our analysis does not go through.
To resolve this issue, let us focus on the bad scenario where $f(p)$ is sufficiently far from $c_i^*$, i.e., $\|f(p) - c_i^*\| \gg \max\{\|p - c_i^*\|, r_{\ell_i}\}$.
This implies $f(p)$ is also far from $p$.
We further show that $\|Gp - Gf(p)\| \gg \|p - c_i^*\|$ by careful analysis of $G$'s randomness.
On the other hand, we have $\|p - C(p)\| \leq O(\|p - c_i^*\|)$ by \eqref{eqn:key_observation}.
Therefore, we can directly lower bound $\|Gp - Gf(p)\|$ by $\|p - C(p)\|$ in this case.

\section{Preliminaries}

Consider a point set $P \subset \R^d$. For every $x\in \R^d$, denote by $P(x)$ the point in $P$ closest to $x$ and $\dist(x, P) :=\|x- P(x)\|$ (recall that throughout $\|\cdot\|$ is the Euclidean norm).
Denote $\diam(P) := \max\{\dist(p, q) \colon p, q \in P\}$ as the \emph{diameter} of $P$.
For $x \in \R^d$ and $r > 0$, denote by $B(x, r) := \{y \in \R^d \colon |x-y\| \leq r\}$ the \emph{ball} centered at $x$ with radius $r$.
Recall that
for $k \in \mathbb{N}$ and $z \geq 1$, the $(k, z)$-clustering cost of $P$ w.r.t. center set $C \subset \R^d, |C| \leq k$ is $\cost_k^z(P, C) := \sum_{p \in P} \dist(p, C)^z$.
The optimal discrete $(k, z)$-clustering cost of $P$ w.r.t. a candidate center set $Q \subset \R^d$ is denoted by $\opt_k^z(P, Q) := \min_{C \subseteq Q, |C| \leq k} \cost_k^z(P, C)$,
and by $\opt(P, Q)$ for short when $k, z$ are clear from the context.
Denote $\opt(P) := \opt(P, P)$ and $\optcont(P) := \opt(P, \R^d)$ for simplicity.

We use the following generalized triangle inequalities.
\begin{lemma}[Generalized triangle inequalities~\cite{MakarychevMR19}]
    \label{lemma:triangle}
    Let $(X, \dist)$ be a metric space. 
    Then for every $z \geq 1$, $\epsilon \in (0, 1)$ and $p, q, r \in X$, 
    \begin{align*}
        &\dist(p, q)^z \geq (1 - z \epsilon) \dist(p, r)^z - \epsilon^{-z} \dist(q, r)^z. \\
        &\dist(p, q)^z \leq (1 + \epsilon)^{z-1} \dist(p, r)^z + \left(\frac{1 + \epsilon}{\epsilon}\right)^{z-1} \dist(q, r)^z.
    \end{align*}
\end{lemma}

\subsection{Doubling dimension and nets}
\begin{definition}[Doubling dimension~\cite{GKL03}]
    \label{def:ddim}
    The \emph{doubling dimension} of a set $P\subseteq \R^d$, denoted $\ddim(P)$,
    is the minimum $m > 0$, such that $\forall r > 0$, every ball in $P$ with radius $r$ can be covered by at most $2^m$ balls of radius $r/2$.
\end{definition}

Our proof uses $\rho$-nets for doubling sets, whose definition and key properties are described here.

\begin{definition}[$\rho$-net]
    \label{def:net}
    Let $P\subseteq \R^d$
    and $\rho > 0$.
    A subset $N \subseteq P$ is called a \emph{$\rho$-packing} of $P$ if $\forall u, v \in N, \ \|u- v\| > \rho$.
    The subset $N$ is called a \emph{$\rho$-covering} of $P$ if $\forall x \in P$, there exists $u \in N$ such that $x \in B(u, \rho)$.
    The subset $N$ is called a \emph{$\rho$-net} of $P$ if $N$ is both a $\rho$-packing and $\rho$-covering of $P$.
\end{definition}

\begin{lemma}[Packing property~\cite{GKL03}]
    \label{lemma:packing}
    Let $P\subseteq \R^d$ and $N \subseteq P$ be a $\rho$-packing of $P$. 
    Then $|N| \leq (\diam(P) / \rho)^{O(\ddim(P))}$.
\end{lemma}

\subsection{Dimension reduction}
For simplicity, we only consider random linear maps defined by a matrix of iid Gaussians, which are known to satisfy the JL Lemma~\cite{IndykM98,DasguptaG03}.
\begin{definition}
A Gaussian JL map is a $t\times d$ matrix with i.i.d. entries drawn from $N(0,\tfrac{1}{t})$. 
\end{definition}

Recall the following concentration bound~\cite[Eq. (7)]{IndykN07} (see also \cite[Eq. (5)]{NarayananSIZ21}), from which one can deduce the JL lemma.
\begin{lemma}[{\cite[Eq. (7)]{IndykN07}}]\label{lem:gaussian_concentration}
    Let $x\in \R^d, \eps>0$ and a Gaussian JL map $G\in\R^{t\times d}$. We have
    \[
    \Pr(\|Gx\|\notin (1\pm \eps)\|x\|)\leq \exp(-\eps^2 t/8).
    \]
\end{lemma}

The following two lemmas regard Gaussian JL maps when applied to doubling sets.

\begin{lemma}[{\cite[Lemma 4.2]{IndykN07}}]
\label{lem:IN07_expansion}

There exist universal constants $A_1,A_2>0$ such that
for every subset $P\subset B(\vec{0},1)$ of the Euclidean unit ball in $\R^d$, 
$t>A_1 \cdot \ddim(P) +1, D \geq 10$, and a Gaussian JL map $G\in\R^{t\times d}$, 
\[
\Pr (\exists x\in P, \|Gx\| >D) \leq e^{-A_2 t D^2}.
\]
\end{lemma}

\begin{lemma}[{\cite[Lemma 3.21]{HuangJKY25}}]\label{lem:IN07_variant_contraction}
    There exists universal constants $A_1,A_2,L>1$, such that for every 
    $P\subset \R^d\setminus B(\vec{0},1), \eps>0, t>A_1 \ddim(P)$, and a Gaussian JL map $G\in\R^{t\times d}$, 
    \[
    \Pr(\exists x\in P, \|Gx\|<\tfrac{1}{L})\leq e^{-A_2 t}.
    \]
\end{lemma}

\section{The first upper bound}\label{sec:wth_loglogn}

We prove \Cref{thm:informal_main_1} (a.k.a \Cref{thm:informal_partitions_and_centers}) in this section, formally stated below.
\begin{theorem}\label{thm:k_medoid_forall_centers_partitions}
    Let $\eps>0$, $z \geq 1$ and $d,\ddim,k\in \N$ and a Gaussian JL map $G\in\R^{t\times d}$ with suitable $t=O(z^2 \eps^{-2}(\ddim\log(z/\epsilon)+ \log k + \log\log n))$.
    For every set $P\subseteq\R^d$ with $\ddim(P)\leq \ddim$, with probability at least $2/3$,
    \begin{enumerate}
        \item $\opt_k^z(G(P)) \leq (1 + \epsilon) \opt_k^z(P)$, and
        \item for all centers $C = (c_1,\ldots,c_k)\subseteq P$ and all partitions $\P=(S_1,\ldots,S_k)$ of $P$,\[\cost_k^z(G(\P),G(C))\geq (1-\eps)\cost_k^z(\P,C),\]
        where $\cost_k^z(\P,C)=\sum_{i=1}^k\sum_{p\in S_i}\|p-c_i\|^z$.
    \end{enumerate}
    
\end{theorem}

We use the following lemma to bound the clustering cost of a fixed set of centers and partition of $P$.
The proof is deferred to \Cref{appendix:fixed_centers_partition}.

\begin{restatable}{lemma}{lemmafixedcenterspartitions}
    \label{lemma:fixed_centers_partiton}
    Let $\eps>0$, $z \geq 1$ and $d, k\in \N$ and a Gaussian JL map $G\in\R^{t\times d}$ with suitable $t=O(z^2 \eps^{-2} \log \epsilon^{-1})$.
    For every set $P\subseteq\R^d$, every set of centers $(c_1, \dots, c_k) \subset \R^d$ and every partition $\P = (S_1, \dots, S_k)$ of $P$, with probability at least $9/10$,
    \[\cost_k^z(G(\P),G(C)) \leq (1+\eps)\cost_k^z(\P,C).\]
\end{restatable}

\begin{proof}[Proof of \Cref{thm:k_medoid_forall_centers_partitions}]
    Consider an optimal discrete $k$-median of $P$. Denote by $C^*=\{c_1^*,\ldots,c_k^*\}\subseteq P$ and by $S_1^*,\ldots, S_k^*$ the centers and clusters (respectively) in that solution.
    Applying \Cref{lemma:fixed_centers_partiton} to the optimal center set $C^*$ and the partition $\P^* = (S_1^*, \dots, S_k^*)$,
    we have that with probability at least $9/10$,
    \[\opt(G(P)) \leq \cost(G(\P^*), G(C^*)) \leq (1 + \epsilon) \cost(\P^*, C^*) = (1 + \epsilon) \opt(P), \] 
    concluding the first part of the theorem.
    
    Denote by $r_0$ the largest radius of any cluster $S_i^*$.
    Pick a suitable $m=O(\log n)$ such that $2^m=n^{10}$.
    For $i\in [0,m]$ and $j\in [k]$, set $r_i = r_0/2^i$, and $P_{ij} = S_j\cap B(c_j,r_i)$, i.e., for every cluster, we have a sequence of geometrically decreasing balls.
    Additionally, let $N_i$ be an $\eps^3 r_i$-net of $\cup_j P_{ij}$.
    By \Cref{lemma:packing}, $|N_i|\leq k\eps^{-O(\ddim(P))}$.

    For each $x,y\in \cup_{i\in [0,m]} N_i$, by \Cref{lem:gaussian_concentration}, 
    \[\Pr(\|Gx-Gy\|> (1+\epsilon)\|x-y\|)\leq \exp(-\eps^2 t/8)\leq \frac{\eps^{\Omega(\ddim(P))}}{k^2m^2}.\]
    Thus, by a union bound, w.p. at least $9/10$,
    \begin{equation}\label{eq:JL_kcenters}
        \forall x,y\in \cup_{i\in [0,m]} N_i, \qquad \|Gx-Gy\|\leq (1+\epsilon)\|x-y\|.
    \end{equation}
    Furthermore, for each $i\in [0,m], y\in N_i$, by \Cref{lem:IN07_expansion},
    \[
    \Pr(\exists p\in P\cap B(y,\epsilon^3 r_i) \ s.t. \ \|G(p-y)\|> 10 \eps^3 r_i)\leq \exp(-\Omega(t)).
    \]
    By a union bound, w.p. at least $9/10$,
    \begin{equation}\label{eq:IN_kcenters}
        \forall i\in [0,m], y\in N_i, p\in P\cap B(y,\epsilon^3 r_i), \qquad \|G(p-y)\|\leq 10 \eps^3 r_i.
    \end{equation}
    By another union bound, \Cref{eq:JL_kcenters,eq:IN_kcenters} hold with probability at least $2/3$.
    
    We are now ready to prove the second part of the theorem.
    Let $C=\{c_1,\ldots,c_k\}\subseteq P$ and let a partition $\P=(S_1,\ldots,S_k)$ of $P$.
    For every $p\in P$ we denote by $u_p$ the nearest net-point to $p$ in the level such that $P_i\setminus P_{i+1}$ contains $p$, and the radius of that level is denoted $r_p$.
    Denote by $f(p)$ the center in $C$ assigned to $p$ according to the partition $\P$.
    Recall that $C^*(p)$ is a point in $C^*$ that is nearest to $p$.
    Observe that
    \[
    \sum_{p \in P} r_p^z \leq n\cdot \left(\frac{r_0}{n^{10}}\right)^z +\sum_{j=1}^k \sum_{i=0}^{m-1} \sum_{p \in P_{i,j}\setminus P_{i+1,j}} (2\|p-c_j^*\|)^z 
    = O(2^z) \cdot \opt(P),
    \]
    and
    \begin{align}
        (\tfrac{1}{2} r_{f(p)})^z 
        &\leq  \|f(p) - C^*(f(p))\|^z && \text{by definition} \nonumber\\
        &\leq \|f(p) - C^*(p)\|^z && \text{$C^*(f(p))$ is nearest to $f(p)$ from $C^*$} \nonumber\\
        &\leq 2^{z-1} \|p-f(p)\|^z + 2^{z-1} \|p - C^*(p)\|^z && \text{by \Cref{lemma:triangle}.} 
        \label{eq:ineq_rad_p_vs_rad_c_p}
    \end{align}
    Therefore,
    \begin{align*}
        & \quad \cost(G(\P), G(C)) \\
        & \equiv \sum_{p \in P} \|Gp - Gf(p)\|^z \\
        & \geq \sum_{p \in P} (1 - z \epsilon) \|Gu_p - Gu_{f(p)}\|^z - \epsilon^{-z} \|Gp - Gu_p\|^z - \epsilon^{-z} \|Gf(p) - Gu_{f(p)}\|^z && \text{by \Cref{lemma:triangle}}\\
        & \geq \sum_{p \in P} (1 - z \epsilon)(1 - \epsilon)^z \|u_p - u_{f(p)}\|^z - \epsilon^{-z} (10 \epsilon^3 r_p)^z - \epsilon^{-z} (10 \epsilon^3 r_{f(p)})^z &&\text{by \eqref{eq:JL_kcenters} and \eqref{eq:IN_kcenters}}\\
        & \geq \sum_{p \in P} (1 - z \epsilon)^2 (1 - \epsilon)^z \|p - f(p)\|^z - O(\epsilon)^z r_p^z - O(\epsilon)^z r_{f(p)}^z &&\text{by \Cref{lemma:triangle}} \\
        & \geq \sum_{p \in P} (1 - 3z \epsilon) \|p - f(p)\|^z 
        - O(\epsilon)^z r_p^z - O(\epsilon)^z 2^{2z-1}(\|p - f(p)\|^z + \|p - C^*(p)\|^z) &&\text{by \eqref{eq:ineq_rad_p_vs_rad_c_p}} \\
        & \geq (1 - O(z \epsilon)) \cost(\P, C) - O(\epsilon) \cdot \opt(P).
    \end{align*}
    Rescaling $\eps \to \epsilon/z$ concludes the proof.
\end{proof}

\section{General candidate centers}
We now consider a generalization of \Cref{thm:k_medoid_forall_centers_partitions}, to the setting where the centers are from a (possibly different than the input) candidate set $Q$.
Unfortunately, to obtain multiplicative contraction in this setting, we have to pay $\Theta(\epsilon^{-2} \log |Q|)$ in the target dimension.
We prove the upper bound below, and the matching lower bound is provided in \Cref{thm:lowerbound_logn_medoid_forall_c_candidate}. 

\begin{theorem}\label{thm:k_medoids_for_all_partitions_c_candidate_multi}
    Let $\eps>0$, $z \geq 1$ and $d, k, s\in \N$ and a Gaussian JL map $G\in\R^{t\times d}$ with suitable $t=O(z^2 \eps^{-2}(\log s + z \log(z / \epsilon)))$.
    For every set $P \subseteq\R^d$ and every candidate center set $Q \subseteq \R^d$ with $|Q| = s \geq k$,
    with probability at least $2/3$,
    \begin{enumerate}
        \item $\opt_k^z(G(P), G(Q)) \leq (1 + \epsilon) \opt_k^z(P, Q)$, and
        \item for every $C=(c_1,\ldots,c_k)\subseteq Q$ and every partition $\P=(S_1,\ldots,S_k)$ of $P$, 
        \[
        \cost_k^z(G(\P), G(C)) \geq (1 - \epsilon) \cost_k^z(\P, C),
        \]
        where $\cost_k^z(\P,C)=\sum_{i=1}^k\sum_{p\in S_i}\|p-c_i\|^z$.
    \end{enumerate}
\end{theorem}

The proof of \Cref{thm:k_medoids_for_all_partitions_c_candidate_multi} uses the following lemma, whose proof is provided in \Cref{sec:proof_lem_preserve_sum}.

\begin{restatable}{lemma}{lemmapreservesum}
    \label{lemma:preserve_sum}
    There exists universal constant $A_2>1$, such that
    for every $P\subset \R^d,\eps>0, z \geq 1, k\in \N$,  $c\in \R^d$, and a Gaussian JL map $G\in\R^{t\times d}$, with probability $1-\epsilon^{-O(z)} k^2 e^{-A_2 \eps^{2}t}$,
    \[
    \forall P'\subseteq P, \qquad \sum_{p\in P'} \|Gp-Gc\|^z \geq (1- \eps)^{3z} \sum_{p\in P'} \|p-c\|^z - \frac{\eps}{k^2} \optcont_k^z(P),
    \]
    where $\optcont_k^z(P)$ is the optimal continuous $(k, z)$-clustering value of $P$.
\end{restatable}

\begin{remark}
    There is a similar statement in \cite{MakarychevMR19}, but w.r.t. the optimal center of $P'$.
    In contrast, here the center is fixed.
\end{remark}

\begin{proof}[Proof of \Cref{thm:k_medoids_for_all_partitions_c_candidate_multi}]
    The first guarantee is the same as \Cref{thm:k_medoid_forall_centers_partitions}, so we omit its proof and focus on the second guarantee.
    By \Cref{lemma:preserve_sum} and a union bound over $Q$, we have that with probability $1 - s \cdot \epsilon^{-O(z)} k^2 e^{-A_2 \epsilon^2 t} \geq 2/3$,
    all centers $c \in Q$ satisfy that
    \begin{equation}\label{eqn:sum_is_preserved}
        \forall P'\subseteq P, \qquad \sum_{p\in P'} \|Gp-Gc\|^z \geq (1- \eps)^{3z} \sum_{p\in P'} \|p-c\|^z - \frac{\eps}{k^2} \optcont_k^z(P).
    \end{equation}

    Consider an arbitrary center set $C = (c_1, c_2, \dots, c_k) \subseteq Q$ and a partition $\P = (S_1, S_2, \dots, S_k)$ of $P$.
    For $i \in [k]$, pick $c = c_i$ and $P' = S_i$ in \eqref{eqn:sum_is_preserved}; we have
    \begin{align*}
        \forall i \in [k], \qquad \sum_{p \in S_i} \|Gp - Gc_i\|^z \geq (1 - \epsilon)^{3z} \sum_{p \in S_i} \|p - c_i\|^z - \frac{\epsilon}{k^2} \optcont_k^z(P).
    \end{align*}
    Summing over $i \in [k]$, we obtain 
    \begin{align*}
        \cost_k^z(G(\P), G(C)) 
        &= \sum_{i = 1}^k \sum_{p \in S_i} \|Gp - Gc_i\|^z \\
        & \geq (1 - \epsilon)^{3z} \sum_{i = 1}^k \sum_{p \in S_i} \|p - c_i\|^z - \epsilon \optcont_k^z(P) \\
        & \geq (1 - \epsilon)^{3z} \sum_{i = 1}^k \sum_{p \in S_i} \|p - c_i\|^z - \epsilon \opt_k^z(P, Q) \\
        & \geq (1 - 3z \epsilon) \cost_k^z(\P, C) - \epsilon \opt_k^z(P, Q) \\
        & \geq (1 - O(z \epsilon)) \cost_k^z(\P, C).
    \end{align*}
    Rescaling $\epsilon$ by a factor of $1/z$ completes the proof.
\end{proof}

To bypass the $O(\epsilon^{-2} \log |Q|)$ barrier in the target dimension, we consider relaxed contraction, and prove the following.

\begin{theorem}\label{thm:k_medoids_for_all_partitions_c_candidate}
    Let $\eps>0$, $z \geq 1$ and $d,\ddim,k\in \N$ and a Gaussian JL map $G\in\R^{t\times d}$ with suitable $t=O(z^2 \eps^{-2}(\ddim\log(z/\eps) + \log k + \log \log \alpha + \log\log n))$.
    For every $n$-point set $P \subseteq\R^d$ and every candidate center set $Q \subseteq \R^d$ with $\ddim(P \cup Q)\leq \ddim$, with probability at least $2/3$,
    \begin{enumerate}
        \item $\opt_k^z(G(P), G(Q)) \leq (1 + \epsilon) \opt_k^z(P, Q)$, and
        \item for every $C=(c_1,\ldots,c_k)\subseteq Q$ and every partition $\P=(S_1,\ldots,S_k)$ of $P$, 
        \[
        \cost_k^z(G(\P), G(C)) \geq \min\{\alpha \cdot \opt_k^z(P, Q), (1 - \epsilon) \cost_k^z(\P, C)\},
        \]
        where $\cost_k^z(\P,C)=\sum_{i=1}^k\sum_{p\in S_i}\|p-c_i\|^z$.
    \end{enumerate}
\end{theorem}

\begin{proof}
    The first guarantee is the same as \Cref{thm:k_medoid_forall_centers_partitions}, so we omit its proof and focus on the second guarantee.
    Consider an optimal discrete $k$-median of $P$. Denote by $C^*=\{c_1^*,\ldots,c_k^*\}\subseteq P$ and by $S_1^*,\ldots, S_k^*$ the centers and clusters (respectively) in that solution.
    Denote $r_0 := \opt_k^z(P, Q)^{1/z}$.
    Pick a suitable $m=O(\log n)$ such that $2^m=n^{10}$.
    Denote $L$ as the same (sufficiently large) constant in \Cref{lem:IN07_variant_contraction}.
    For $i\in [-\log(10L \alpha), m]$ and $j\in [k]$, set $r_i = r_0/2^i$, and $P_{ij} = S_j^* \cap B(c_j,r_i)$, i.e., for every cluster, we have a sequence of geometrically decreasing balls.
    Additionally, let $N_i$ be an $\eps^3 r_i$-net of $\cup_j P_{ij}$.
    As in the proof of \Cref{thm:k_medoid_forall_centers_partitions}, we have $\sum_{p\in P}r_p^z=O(2^z) \opt$.

    By \Cref{lem:gaussian_concentration,lem:IN07_expansion,lem:IN07_variant_contraction} and a union bound, the following hold with probability at least $2/3$,
    \begin{align}
        &\forall x,y\in \cup N_i,  &&\|Gx-Gy\|\leq (1+\epsilon)\|x-y\| \label{eq:JL_kcenters_eps} \\
        &\forall i\in [-\log(10 L \alpha), m], y\in N_i, p\in P\cap B(y,\epsilon^3 r_i), &&\|G(p-y)\|\leq 10 \eps^3 r_i\label{eq:IN_kcenters_eps} \\
        &\forall i \in [k], \forall y \in (P \cup Q) \cap B(c_i^*, r_0), && \|Gy - G c_i^*\| \leq 10  r_0 \label{eq:IN_kcenters_expansion} \\
        &\forall i \in [k], \forall y \in (P \cup Q) \setminus B(c_i^*, 10L \alpha \cdot r_0), && \|Gy - G c_i^*\| \geq 10\alpha \cdot r_0\label{eq:IN_kcenters_contraction}
    \end{align}
    \Cref{eq:JL_kcenters_eps,eq:IN_kcenters_eps} are the same as \Cref{eq:JL_kcenters,eq:IN_kcenters}, and hold with probability at least $9/10$. 
    \Cref{eq:IN_kcenters_expansion,eq:IN_kcenters_contraction} each hold w.p. $9/10$ directly by \Cref{lem:IN07_expansion,lem:IN07_variant_contraction}, respectively. A union bound yields the desired success probability $>2/3$.

    We are now ready to prove the theorem.
    Let $C=\{c_1,\ldots,c_k\}\subseteq Q$ and let a partition $\P=(S_1,\ldots,S_k)$ of $P$.
    For $p \in P$, denote $f(p) \in C$ to be the center to which $p$ is assigned.
    Consider the following cases.

    \paragraph{Case 1, $\exists j \in [k]$, s.t. $\|c_j - C^*(c_j)\| \geq 10 L \alpha \cdot r_0$ and $S_j\not\subseteq \{c_j\}$.}
    By assumption, there exists a point $p \in S_j, p\neq c_j$. 
    Then $\|c_j - C^*(p)\| \geq \|c_j - C^*(c_j)\| \geq 10L \alpha \cdot r_0$.
    By \eqref{eq:IN_kcenters_contraction}, $\|G c_j - G C^*(p)\| \geq 10\alpha r_0$.
    On the other hand, $\|p - C^*(p)\| \leq r_0$.
    By \eqref{eq:IN_kcenters_expansion}, $\|Gp - GC^*(p)\| \leq 10 r_0$. Hence, 
    \begin{align*}
        \cost(G(\P), G(C)) &\geq \|Gp - Gc_j\|^z \\
        & \geq (\|G c_j - G C^*(p)\| - \|Gp - GC^*(p)\|)^z \\
        &\geq (10 \alpha r_0 - 10 r_0)^z \\ 
        &> \alpha r_0^z \\
        &= \alpha \opt(P, Q).
    \end{align*}

    \paragraph{Case 2, $\forall j \in [k], \ \|c_j - C^*(c_j)\| \leq 10 L \alpha \cdot r_0$ or $S_j\subseteq \{c_j\}$.}
    Without loss of generality, we can assume that for all $j\in [k], S_j\not\subseteq \{c_j\}$. That is since whenever $S_j\subseteq \{c_j\}$, we have $\cost(S_j,c_j)=\cost(G(S_j),Gc_j)=0$.
    
    Therefore, every center in $C$ is covered by the union of nets $\bigcup_i N_i$. 
    For every $p\in P \cup Q$ we denote by $u_p$ the nearest net-point to $p$ in the level such that $P_i\setminus P_{i+1}$ contains $p$, and the radius of that level is denoted $r_p$.
    Same as the proof of \Cref{thm:k_medoid_forall_centers_partitions}, we are able to establish \eqref{eq:ineq_rad_p_vs_rad_c_p} for every $r_{f(p)}$. Then 
    \begin{align*}
        & \quad \cost(G(\P), G(C)) \\
        & \equiv \sum_{p \in P} \|Gp - Gf(p)\|^z \\
        & \geq \sum_{p \in P} (1 - z \epsilon) \|Gu_p - Gu_{f(p)}\|^z - \epsilon^{-z} \|Gp - Gu_p\|^z - \epsilon^{-z} \|Gf(p) - Gu_{f(p)}\|^z && \text{by \Cref{lemma:triangle}}\\
        & \geq \sum_{p \in P} (1 - z \epsilon)(1 - \epsilon)^z \|u_p - u_{f(p)}\|^z - \epsilon^{-z} (10 \epsilon^3 r_p)^z - \epsilon^{-z} (10 \epsilon^3 r_{f(p)})^z &&\text{by \eqref{eq:JL_kcenters_eps} and \eqref{eq:IN_kcenters_eps}}\\
        & \geq \sum_{p \in P} (1 - z \epsilon)^2 (1 - \epsilon)^z \|p - f(p)\|^z - O(\epsilon)^z r_p^z - O(\epsilon)^z r_{f(p)}^z &&\text{by \Cref{lemma:triangle}} \\
        & \geq \sum_{p \in P} (1 - 3z \epsilon) \|p - f(p)\|^z 
        - O(\epsilon)^z r_p^z - O(\epsilon)^z 2^{2z-1}(\|p - f(p)\|^z + \|p - C^*(p)\|^z) &&\text{by \eqref{eq:ineq_rad_p_vs_rad_c_p}} \\
        & \geq (1 - O(z \epsilon)) \cost(\P, C) - O(\epsilon) \cdot \opt(P).
    \end{align*}
    Rescaling $\eps \to \epsilon/z$ concludes the proof.
\end{proof}

\subsection{Proof of \Cref{lemma:preserve_sum}}\label{sec:proof_lem_preserve_sum}

\lemmapreservesum*
Our proof of \Cref{lemma:preserve_sum} is based on \cite{danos2021thesis}, and is by reducing to the central symmetric case.
We say a point set $X \subset \R^d$ is \emph{central symmetric} with center $c \in \R^d$, if for every point $x \in X$, it holds $2c - x \in X$.
The following lemma shows that the (continuous) $1$-median center of a central symmetric point set coincides with its center of symmetry.

\begin{lemma}\label{lemma:central_symmetric}
    Let $z \geq 1$ and $X \subset \R^d$ be a central symmetric point set centered at point $c \in \R^d$. 
    Then $c$ is an optimal (continuous) $(1, z)$-clustering center of $X$.
\end{lemma}

\begin{proof}
    For $x \in X$, denote $x^- := 2c - x$.
    Let $c^* \in \R^d$ be an optimal $1$-median center of $X$. Then 
    \begin{align*}
        \cost(X, c^*) &= \sum_{x \in X} \|x - c^*\|^z \\
        &= \frac{1}{2} \sum_{x \in X} (\|x - c^*\|^z + \|x^- - c^*\|^z) \\
        &\geq \frac{1}{2^z} \sum_{x \in X} (\|x - c^*\| + \|x^- - c^*\|)^z  \\
        &\geq \frac{1}{2^z} \sum_{x \in X}\|x - x^-\|^z.
    \end{align*}
\end{proof}

Denote $\optcont_k^z(P)$ as the optimal continuous $(k, z)$-clustering value of $P$.
The following lemma is a restatement of~\cite[Theorem 3.4]{MakarychevMR19}.

\begin{lemma}[{\cite[Theorem 3.4]{MakarychevMR19}}]
\label{lemma:1median_mmr}
    Consider a point set $X \subset \R^d$. 
    Let $G$ be a random linear map and $C$ be a random subset of $X$ (which may depend on $G$).
    Then with probability at least $1 - O(\epsilon^{-O(z)}k^2 e^{-\Omega(\epsilon^{-2} t)})$,
    \begin{equation*}
        \optcont_1^z(G(C)) \geq (1 - \epsilon)^{3z} \optcont_1^z(C) - \frac{\epsilon}{k^2} \optcont_k^z(X).
    \end{equation*}
\end{lemma}

\begin{proof}[Proof of \Cref{lemma:preserve_sum}]

Let $\tilde{P} \subseteq P$ be a subset that maximizes 
\[(1 - \epsilon)^{3z} \sum_{p \in \tilde{P}} \|p - c\|^z - \sum_{p \in \tilde{P}} \|Gp - Gc\|^z.\]
Then $\tilde{P}$ is a random subset of $P$ that depends on $G$.
Denote $X := P \cup \{2c - p \colon p \in P\}$ and $\tilde{X} := \tilde{P} \cup \{2c - p \colon p \in \tilde{P}\}$. 
Then both $X$ and $\tilde{X}$ are symmetric with center $c$.
By \Cref{lemma:1median_mmr}, with $k'=2k$, with probability $1 - O(\epsilon^{-O(z)} k^2 e^{- \Omega(\epsilon^2 t)})$,
\[
\optcont_1^z(G(\tilde X)) \geq (1 - \epsilon)^{3z} \optcont_1^z(\tilde X) - \frac{\epsilon}{4k^2} \optcont_{2k}^z(X).
\]
By the linearity of $G$, $G(\tilde{X})$ is symmetric with center $Gc$.
Hence, by \Cref{lemma:central_symmetric}, we have
\begin{align*}
    \sum_{x \in \tilde{X}} \|Gx - Gc\|^z \geq (1 - \epsilon)^{3z} \sum_{x \in \tilde X} \|x - c\|^z - \frac{\epsilon}{4k^2} \optcont_{2k}^z(X).
\end{align*}
Note that $\optcont_{2k}^z(X) \leq 2 \optcont_k^z(P)$.
Then 
\begin{align*}
    \sum_{p \in \tilde{P}} \|Gp - Gc\|^z \geq (1 - \epsilon)^{3z} \sum_{p \in \tilde P} \|p - c\|^z - \frac{\epsilon}{2 k^2} \optcont_{k}^z(P).
\end{align*}
By the definition of $\tilde{P}$, we conclude that with probability $1 - O(\epsilon^{-O(z)} k^2 e^{- \Omega(\epsilon^2 t)})$,
\begin{align*}
    \forall P' \subseteq P, \qquad \sum_{p \in P'} \|Gp - Gc\|^z \geq (1 - \epsilon)^{3z} \sum_{p \in P'} \|p - c\|^z - \frac{\epsilon}{2 k^2} \optcont_{k}^z(P).
\end{align*}

\end{proof}
\section{Improved upper bound: Removing the $\log \log n$ term}
\label{sec:forall_centers_relaxed}

We prove \Cref{thm:informal_main_2} in this section. In fact, we prove the following for the more general candidate centers setting, and get \Cref{thm:informal_main_2} by setting $Q=P$.

\begin{theorem}\label{thm:k-medoid_no_loglogn_improved}
    Let $0<\eps<\tfrac{1}{2}$, $z \geq 1$, $\alpha > 2$ and $d,\ddim\in \N$ and a Gaussian JL map $G\in\R^{t\times d}$ with suitable 
    $t = O(z^2 \epsilon^{-2} (\ddim \log(z/\eps) + z \log(z/\eps) + \log k + \log \log \alpha))$, 
    the following holds.
    For every $P, Q\subseteq\R^d$ with $\ddim(P \cup Q)\leq \ddim$, 
    with probability at least $2/3$,
    \begin{enumerate}
        \item $\opt_k^z(G(P), G(Q)) \leq (1 + \epsilon) \opt_k^z(P, Q)$, and
        \item $\forall C \subseteq Q, |C| = k$,
        \[
        \cost_k^z(G(P),G(C)) \geq \min\{\alpha \cdot \opt_k^z(P, Q), (1-\eps)\cost_k^z(P,C)\}.
        \]
    \end{enumerate}
\end{theorem}

\medskip

Consider an optimal discrete $k$-median of $P$ w.r.t. candidate center set $Q$. 
Denote by $C^*=\{c_1^*,\ldots,c_k^*\}\subseteq Q$ and by $S_1^*,\ldots, S_k^*$ the centers and clusters (respectively) in that solution.
Denote $r_0 := \opt(P, Q)^{1/z}$.
For $\ell \in \mathbb{N}$ and $i\in [k]$, set $r_\ell = r_0/2^\ell$, and $P_\ell^i = S_i^* \cap B(c_i^*,r_\ell)$, i.e., for every cluster, we have a sequence of geometrically decreasing balls.
Additionally, let $N_\ell^i$ be an $\eps^3 r_\ell$-net of $(P \cup Q) \cap B(c_i^*, r_{\ell - \log \epsilon^{-1}})$.
Let $N_\ell := \bigcup_i N_\ell^i$.

For $p \in P \cup Q$, recall $C^*(p)$ is the closest center to $p$ in $C^*$.
Let $j_p \in \mathbb{N}$ be the level satisfying $r_{j_p + 1} \leq \|p - C^*(p)\| \leq r_{j_p}$.
Denote $r_p := r_{j_p}$ for simplicity.
We have the following claim.

\begin{lemma}\label{lemma:sum_of_rp}
    $\sum_{p \in P} r_p^z \leq 2^z \opt_k^z(P, Q)$.
\end{lemma}

For $C \subseteq Q$ and $p \in P$, recall we denote by $C(p)$ the point closest to $p$ in $C$.
We have the following lemma that upper bounds the distance from $C(p)$ to $C^*(p)$ (and also the distance from $C(p)$ to $p$).

\begin{lemma}\label{lemma:bound_dist_p_cp}
    Let $C \subseteq Q$.
    Then for every $i \in [k]$ and $p \in S_i^*$, it holds that $\|C(p) - c_i^*\| \leq 4 \max\{r_p, \|c_i^* - C(c_i^*)\|\}$.
\end{lemma}

\begin{proof}
\begin{align*}
    \|C(p)-c_i^*\|&\leq \|C(p)-p\|+\|p-c_i^*\| && \text{by triangle inequality}\\
    &\leq \|C(c_i^*)-p\|+\|p-c_i^*\| && \text{$C(p)$ is the point closest to $p$ in $C$} \\
    &\leq \|C(c_i^*)-c_i^*\| + \|c_i^*-p\|+\|p-c_i^*\| && \text{by triangle inequality}\\
    &\leq 4 \max\{r_p, \|c_i^* - C(c_i^*)\|\}. 
\end{align*}
\end{proof}

\begin{proof}[Proof of \Cref{thm:k-medoid_no_loglogn_improved}]
The first guarantee is the same as \Cref{thm:k_medoid_forall_centers_partitions}, so we omit its proof and focus on the second guarantee.
For a generic solution $C \subseteq Q, |C| = k$, denote $C = \{c_1, c_2, \dots, c_k\}$.
Denote $f(p) := G^{-1}(GC(Gp))$, i.e., $f(p)$ is a center in $C$ realizing $\dist(Gp,G(C))$.
For $j \in [k]$, denote $S_j := \{p \in P \colon f(p) = c_j\}$ as the cluster induced by $c_j$.

For every $i \in [k]$, define the ``threshold level'' of cluster $i$ as 
\begin{equation}\label{eqn:threshold_level}
    \ell_i := \max\{\ell \colon |P_\ell^i| \cdot r_\ell^z > \alpha \opt(P, Q)\}.
\end{equation}
We also define the $i$-th ``buffer'' as $I_i := [\ell_i - \log(2000 L^2), \ell_i + \log(\alpha k)]$, where $L$ is the (sufficiently large) constant in \Cref{lem:IN07_variant_contraction}.

For $0 \leq \ell \leq m$, denote random variable $\beta_\ell$ to be the minimum real, such that $\forall u, v \in N_\ell, \ \|Gu - Gv\| \geq (1 - \epsilon - \beta_\ell \epsilon) \|u - v\|$.
Denote random variable $\gamma_\ell$ to be the minimum real, such that 
$\forall u \in N_\ell, v \in B(u, \epsilon^3 r_\ell), \ \|Gu - Gv\| \leq \gamma_\ell \epsilon^3 r_\ell$.
For $p \in P \cup Q$, write $\beta_p := \beta_{j_p}$ and $\gamma_p := \gamma_{j_p}$ for simplicity.

In the following lemma, we define our good events and bound their success probability.
The proof is deferred to \Cref{sec:prob_good_events}.
\begin{restatable}{lemma}{lemmagoodevents}
\label{lemma:good_events}
With probability at least $0.99$, the following events happen simultaneously.
\begin{enumerate}[(a)]
    \item \label{item:bound_a}
    $\sum_{p \in P} \beta_p r_p^z \leq e^{-\Omega(\epsilon^2 t)} \cdot \opt(P, Q)$, 
    and $\sum_{p \in P} \gamma_p^z r_p^z \leq 10^z \cdot O(\opt(P, Q))$.
    \item \label{item:bound_c}
    $\forall i \in [k], \forall \ell \in I_i, \ \forall u \in N_{\ell}, v \in B(u, \epsilon^3 r_{\ell}), \ \|Gu - Gv\| \leq 10 \epsilon^3 r_{\ell}$. 
    \item \label{item:bound_d}
    $\forall i \in [k], \forall \ell \in I_i $, every net point $u \in N_\ell$ satisfies that $\forall P' \subseteq P$, 
    \[\sum_{p \in P'} \|Gp - Gu\|^z \geq (1 - \epsilon)^{3z} \sum_{p \in P'} \|p - u\|^z - \frac{\epsilon}{k^2} \opt(P, Q). \]
    \item \label{item:bound_e}
    $\forall i \in [k], \forall y \in B(c_i^*, 40L \cdot r_{\ell_i}), \ \|Gy - Gc_i^*\| \leq 400 L \cdot r_{\ell_i}$.
    \item \label{item:bound_f}
    $\forall i \in [k], \forall y \in (P \cup Q) \setminus B(c_i^*, 2000 L^2 \cdot r_{\ell_i}), \ \|Gy - Gc_i^*\| > 2000 L \cdot r_{\ell_i}$.
    \item \label{item:bound_g}
    For $p \in P$, denote by random variable $\xi_p := \min_{y \colon \|y - p\| > 9L \cdot r_{\ell_i}} \|Gy - Gp\|$.
    Then 
    $\forall i \in [k]$,
    \[\sum_{p \in P_{\ell_i}^i} \xi_p^z > \alpha \opt(P, Q).\]
    \item \label{item:bound_h}
    For $p \in P$, denote $\eta_p := \min_{y \colon \|y - p\| > 9L r_p} \|Gy - Gp\|$.
    Then
    $\forall i \in [k]$,
    \[\sum_{p \in S_i^*} \max\{0, (9 r_p)^z - \eta_p^z\}
    \leq e^{-\Omega(t)} \cdot \sum_{p \in S_i^*} r_p^z.\]
\end{enumerate}
\end{restatable}

The proof proceeds by a careful case analysis.

\paragraph{Case 1, one cluster with no cover: $\max_{1 \leq i \leq k}\{\|c_i^* - C(c_i^*)\| - 10L \cdot r_{\ell_i}\} > 0$.}
Then there exists $i \in [k]$, such that $\|c_i^* - C(c_i^*)\| > 10L \cdot r_{\ell_i}$.
Intuitively, this means all points in $C$ are far away from $c_i^*$.
Write \begin{align}
    \cost(G(P), G(C)) \geq \cost(G(P_{\ell_i}^i), G(C))
    = \sum_{p \in P_{\ell_i}^i} \|Gp - Gf(p)\|^z. \label{eqn:cost_GP_GC_lb_new}
\end{align}

Note that for every $p \in P_{\ell_i}^i$, 
\begin{align*}
    \|p - f(p)\| &\geq 
    \|p - C(p)\| \\ 
    &\geq \|C(p) - c_i^*\| - \|p - c_i^*\| \\
    &\geq \|c_i^* - C(c_i^*)\| - \|p - c_i^*\| \\
    &> 10L \cdot r_{\ell_i} - r_p \\
    &\geq 9L \cdot r_{\ell_i}.
\end{align*}

Therefore, $\|Gp - Gf(p)\| \geq \xi_p$.
Combining with \eqref{eqn:cost_GP_GC_lb_new} yields 
\begin{align*}
    \cost(G(P), G(C)) 
    \geq \sum_{p \in P_{\ell_i}^i} \|Gp - Gf(p)\|^z
    \geq \sum_{p \in P_{\ell_i}^i} \xi_p^z
    > \alpha \opt(P, Q).
\end{align*}
where the last inequality follows from event \ref{item:bound_g}.

\paragraph{Case 2, $\max_{1 \leq i \leq k}\{\|c_i^* - C(c_i^*)\| - 10L \cdot r_{\ell_i}\} \leq 0$.}
Then for every $i \in [k]$, $\|c_i^* - C(c_i^*)\| \leq 10L \cdot r_{\ell_i}$, which intuitively means every center in $C^*$ has a nearby neighbor in $C$.

\paragraph{Comparing ``fake'' centers to optimal centers.}
Let $i\in [k]$. For every $p\in S^*_i$, we consider the distance of $p$'s ``fake'' center $f(p)$ (recall, $Gf(p)$ realizes $\dist(Gp,G(C))$) from $p$'s optimal center $c^*_i$.
There are three ranges we consider for $\|f(p)-c^*_i\|$.

Define $R_i := \{p \in S_i^* \colon r_{\ell_i} / (\alpha k) \leq \|f(p) - c_i^*\| \leq 2000 L^2 \cdot r_{\ell_i}\}$, and denote $R := \bigcup_{i = 1}^k R_i$ (called ``the middle range''). 
Moreover, define $T_i:=\{p \in S_i^* \colon \|f(p) - c_i^*\| \leq r_{\ell_i} / (\alpha k)\}$, and denote $T := \bigcup_{i = 1}^k T_i$ (called ``the close range'').

\paragraph{Case 2.1, the middle range $p \in R$.}
Let us first lower bound $\|Gp - Gf(p)\|$ for $p \in R$.
Assume $C^*(p) = c_i^*$ and $f(p) = c_j$, where $i, j \in [k]$.
Since $p \in R_i$, we can assume $r_{\ell + 1} < \|c_j - c_i^*\| \leq r_\ell$ for some level $\ell \in I_i$.
Let $u_{i, j}$ be the net point in $N_\ell$ closest to $c_j$.
Then
\begin{align*}
    \|Gp - Gf(p)\|^z 
    &\geq (1 - z \epsilon)\|Gp - Gu_{i, j}\|^z - \epsilon^{-z} \|Gc_j - Gu_{i, j}\|^z && \text{by \Cref{lemma:triangle}}\\
    &\geq (1 - z \epsilon)\|Gp - Gu_{i, j}\|^z - \epsilon^{-z} (10 \epsilon^3 r_\ell)^z && \text{by event \ref{item:bound_c}} \\
    & \geq (1 - z \epsilon)\|Gp - Gu_{i, j}\|^z - O(\epsilon)^{2z} \|c_j - c_i^*\|^z \\
    &\geq (1 - z \epsilon)\|Gp - Gu_{i, j}\|^z - O(\epsilon)^{2z} \|p - c_j\|^z - O(\epsilon)^{2z} \|p - c_i^*\|^z &&\text{by \Cref{lemma:triangle}}\\
    &= (1 - z \epsilon)\|Gp - Gu_{i, j}\|^z - O(\epsilon)^{2z} \|p - f(p)\|^z - O(\epsilon)^{2z} \|p - C^*(p)\|^z.
\end{align*}
Summing over $p \in R$, we have 
\begin{align}
    &\qquad \sum_{p \in R} \|Gp - Gf(p)\|^z \notag \\
    &= \sum_{i = 1}^k \sum_{j = 1}^k \sum_{p \in R_i \cap S_j} \|Gp - Gc_j\|^z \notag \\
    &\geq \sum_{i = 1}^k \sum_{j = 1}^k \sum_{p \in R_i \cap S_j} \Big((1 - z \epsilon)\|Gp - Gu_{i, j}\|^z - O(\epsilon)^{2z} \|p - f(p)\|^z - O(\epsilon)^{2z} \|p - C^*(p)\|^z\Big) \notag \\
    &\geq \sum_{i = 1}^k \sum_{j = 1}^k \sum_{p \in R_i \cap S_j} (1 - z\epsilon) \|Gp - Gu_{i, j}\|^z - O(\epsilon)^{2z} \sum_{p \in R} \|p - f(p)\|^z - O(\epsilon)^{2z} \opt(P, Q) \notag \\
    &\text{Applying event \ref{item:bound_d} to net point $u_{i, j}$ and subset $R_i \cap S_j$, we have } \notag \\
    &\geq \sum_{i = 1}^k \sum_{j = 1}^k \left(
    (1 - O(z\epsilon)) \sum_{p \in R_i \cap S_j} \|p - u_{i, j}\|^z - \frac{\epsilon}{k^2} \opt(P, Q)
    \right) - O(\epsilon) \sum_{p \in R} \|p - f(p)\|^z - O(\epsilon) \opt(P, Q) \notag \\
    & \geq (1 - O(z \epsilon)) \sum_{p \in R} \|p - f(p)\|^z - O(\epsilon) \opt(P, Q) \notag\\
    & \geq (1 - O(z \epsilon)) \sum_{p \in R} \|p - C(p)\|^z - O(\epsilon) \opt(P, Q). \label{eqn:p_in_R}
\end{align}

\paragraph{Case 2.2, the close range $p\in T$.}
This is somewhat of a special case of Case 2.1.
Assume $C^*(p) = c_i^*$ and $f(p) = c_j$, where $i, j \in [k]$.
Since $p \in R_i$, we have $\|c_j - c_i^*\| \leq r_{\ell}$ for $\ell=\ell_i+\log(\alpha k)$.
Let $u_{i, j}$ be the net point in $N_\ell$ closest to $c_j$.
we have,
\begin{align*}
    \|Gp - Gf(p)\|^z &\geq (1 - z \epsilon) \|Gp - Gu_{i, j}\|^z - \epsilon^{-z} \|Gc_j - Gu_{i, j}\|^z \\
    &\geq (1 - z \epsilon) \|Gp - Gu_{i, j}\|^z - O(\epsilon)^{2z} r_\ell^z.
\end{align*}
If $p\notin B(c^*_i,r_{\ell_i+1})$, then
\[
r_\ell\leq \tfrac{1}{2k} \|p-c^*_i\|\leq \tfrac{1}{2k}(\|p-c_j\|+\|c_j-c_i^*\|)\leq \tfrac{1}{2k} (\|p-c_j\|+r_\ell).
\]
Rearranging, we obtain $r_\ell\leq \|p-c_j\|$.
Summing over $p \in T$, we have 
\begin{align}
    &\qquad \sum_{p \in T} \|Gp - Gf(p)\|^z \notag \\
    &= \sum_{i = 1}^k \sum_{j = 1}^k \sum_{p \in T_i \cap S_j} \|Gp - Gc_j\|^z \notag \\
    &\geq \sum_{i = 1}^k \sum_{j = 1}^k \sum_{p \in T_i \cap S_j} \Big((1 - z\epsilon) \|Gp - Gu_{i, j}\|^z 
    - O(\epsilon)^{2z} \|p - c_j\|^z \Big) 
    - O(\eps)^{2z} \Big|P\cap B(c^*_i,r_{\ell_{i}+1})\Big|\cdot r_\ell^z   \notag \\
    &\geq \sum_{i = 1}^k \sum_{j = 1}^k \sum_{p \in T_i \cap S_j} \Big((1 - z \epsilon) \|Gp - Gu_{i, j}\|^z - O(\epsilon)^{2z} \|p - c_j\|^z \Big) - O(\eps\opt) && \text{by choice of $\ell_i$}   \notag \\
    &\text{Applying event \ref{item:bound_d} to net point $u_{i, j}$ and subset $T_i \cap S_j$, we have } \notag \\
    &\geq \sum_{i = 1}^k \sum_{j = 1}^k \sum_{p \in T_i \cap S_j} \Big(
        (1-O(z\eps))\|p - u_{i, j}\|^z - \frac{\epsilon}{k^2} \opt - O(\epsilon)^{2z} \|p - c_j\|^z
    \Big) - O(\eps\opt)  \notag \\
    &\geq \sum_{i = 1}^k \sum_{j = 1}^k \sum_{p \in T_i \cap S_j} \Big((1-O(z \eps))\|p - c_j\|^z - \frac{\epsilon}{k^2} \opt\Big) - O(\eps\opt)  \notag \\
    & \geq (1 - O(z \epsilon)) \sum_{p \in T} \|p - C(p)\|^z - O(\epsilon) \opt(P, Q). \label{eqn:p_in_T}
\end{align}

\paragraph{Case 2.3, the far range $p \notin R\cup T$.}
We now consider points $p \in S^*_i\setminus (R\cup T)$, i.e., $\|f(p)-c^*_i\|\geq 2000L^2 r_{\ell_i}$.
Suppose $f(p)=c_j$.
By \ref{item:bound_f}, $\|Gc_j-Gc^*_i\|\geq 2000L r_{\ell_i}$.
\begin{claim}
    In this case, $r_p\geq 10Lr_{\ell_i}$.    
\end{claim}
\begin{proof}
    Assume by contradiction that $r_p< 10Lr_{\ell_i}$.
    By \Cref{lemma:bound_dist_p_cp}, $\|C(p)-c_i^*\|\leq 4\max\{r_p,\|c^*_i-C(c^*_i)\|\} \leq 40 Lr_{\ell_i}$.
    Thus by \ref{item:bound_e}, $Gp,GC(p)\in B(Gc^*_i,400Lr_{\ell_i})$.
    Therefore,
    \[
    \|Gc_j-Gc^*_i\|\leq \|Gc_j-Gp\| + \|Gp-Gc^*_i\| \leq \|GC(p)-Gp\|+ \|Gp-Gc^*_i\|\leq 800Lr_{\ell_i},
    \]
    contradiction.
\end{proof}
Therefore, by \Cref{lemma:bound_dist_p_cp}, $\|C(p)-c_i^*\|\leq 4r_p$ and hence
\begin{equation}\label{eqn:p_Cp_leq_rp}
    \|p - C(p)\| \leq \|p - c_i^*\| + \|C(p) - c_i^*\| \leq 5 r_p.
\end{equation}
On a high level, as can be seen by the claim, we have that both $f(p)$ and $p$ are far from $c^*_i$. 
We split into cases depending which of $p$ or $f(p)$ is farther from $c^*_i$ (up to a constant), as follows.

\paragraph{Case 2.3.1, $p \in S_i^* \setminus (R\cup T)$, and $\|f(p) - c_i^*\| > 10 L r_p$.}
By triangle inequality,  
\[\|p - f(p)\| \geq \|f(p) - c_i^*\| - \|p - c_i^*\| \geq 9L \cdot r_p.\]
By the definition of $\eta_p$, we have $\|Gp - Gf(p)\| \geq \eta_p$.
Therefore,
\begin{align}
    &\qquad \sum_{\substack{p \in S_i^* \setminus (R\cup T) \\ \|f(p) - c_i^*\| > 10L r_p }} \|Gp - Gf(p)\|^z \notag\\
    & \geq \sum_{\substack{p \in S_i^* \setminus (R\cup T) \\ \|f(p) - c_i^*\| > 10L r_p }} \eta_p^z && \text{since $\|p - f(p)\| \geq 9L r_p$} \notag \\
    & \geq \sum_{\substack{p \in S_i^* \setminus (R\cup T) \\ \|f(p) - c_i^*\| > 10L r_p}} (9 r_p)^z - e^{-\Omega(t)} \cdot \sum_{p \in S_i^*} r_p^z
    && \text{by event \ref{item:bound_h}} \notag \\
    & \geq \sum_{\substack{p \in S_i^* \setminus (R\cup T) \\ \|f(p) - c_i^*\| > 10L r_p}} (5 r_p)^z - e^{-\Omega(t)} \cdot \sum_{p \in S_i^*} r_p^z \notag \\
    & \geq \sum_{\substack{p \in S_i^* \setminus (R\cup T) \\ \|f(p) - c_i^*\| > 10L r_p}} \|p - C(p)\|^z - e^{-\Omega(t)} \cdot \sum_{p \in S_i^*} r_p^z
    && \text{by \eqref{eqn:p_Cp_leq_rp}} %
    \label{eqn:low_level_chrge_ro_rp}
\end{align}

\paragraph{Case 2.3.2, $p \in S_i^* \setminus (R\cup T)$, and $\|f(p) - c_i^*\| \leq 10 L r_p$.}

Denote $u_p$ and $u_{f(p)}$ to be the net points in $N_{j_p}$ that are closest to $p$ and $f(p)$, respectively.
Then
\begin{align*}
    &\quad \|Gp - Gf(p)\|^z  \\
    &\geq (1 - 2z \epsilon) \|Gu_p - Gu_{f(p)}\|^z - \epsilon^{-z} \|Gp - Gu_p\|^z - \epsilon^{-z} \|Gf(p) - Gu_{f(p)}\|^z &&\text{by triangle inequality}\\
    & \geq (1 - 2 z \epsilon)(1 - \epsilon - \beta_p \epsilon)^z \|u_p - u_{f(p)}\|^z
    - 2 \epsilon^{-z} (\gamma_p \epsilon^3 r_p)^z 
    &&\text{by definitions of $\beta_p, \gamma_p$}\\
    & \geq (1 - 3z \epsilon - \beta_p z\epsilon) \|u_p - u_{f(p)}\|^z - O(\epsilon)^{2z} \gamma_p^z r_p^z \\
    & \geq (1 - 3z \epsilon - \beta_p z\epsilon) \|p - f(p)\|^z - O(\epsilon)^{2z} r_p^z - O(\epsilon)^{2z} \gamma_p^z r_p^z \\
    &\text{Since $\|p - f(p)\| \leq \|p - c_i^*\| + \|f(p) - c_i^*\| \leq r_p + 10L r_p \leq 20L r_p$, we have} \\
    & \geq (1 - 3z \epsilon) \|p - f(p)\|^z
    - \beta_p z \epsilon \cdot (20L)^z r_p^z 
     - O(\epsilon)^{2z} r_p^z - O(\epsilon)^{2z} \gamma_p^z r_p^z.
\end{align*}
Therefore,
\begin{align}
    & \qquad \sum_{\substack{p \in S_i^* \setminus (R\cup T) \\ \|f(p) - c_i^*\| \leq 10 L r_p}} \|Gp - Gf(p)\|^z \notag \\
    & \geq (1 - 3z \epsilon) \sum_{\substack{p \in S_i^* \setminus (R\cup T) \\ \|f(p) - c_i^*\| \leq 10 L r_p}} \|p - f(p)\|^z
    - z \epsilon (20 L)^z \sum_{p \in S_i^*} \beta_p r_p^z 
    - O(\epsilon)^{2z} \sum_{p \in S_i^*} (1 + \gamma_p^z) r_p^z.
    \label{eqn:high_level_new}
\end{align}

\paragraph{Wrap Up.}
Combining \eqref{eqn:high_level_new} and \eqref{eqn:low_level_chrge_ro_rp}, we have 
\begin{align*}
    \sum_{p \in S_i^* \setminus (R\cup T)} \|Gp - Gf(p)\|^z
    \geq (1 - 3z \epsilon) \sum_{p \in S_i^* \setminus (R\cup T)} \|p - C(p)\|^z
    - z \epsilon (20 L)^z \sum_{p \in S_i^*} \beta_p r_p^z 
    - O(\epsilon)^{2z} \sum_{p \in S_i^*} (1 + \gamma_p^z) r_p^z.
\end{align*}
Summing over $i \in [k]$ yields 
\begin{align}
    &\qquad \sum_{p \in P\setminus (R\cup T)} \|Gp - Gf(p)\|^z \notag \\
    &\geq (1 - 3z \epsilon) \sum_{p \in P \setminus (R\cup T)} \|p - C(p)\|^z
    - z \epsilon (20L)^z \sum_{p \in P} \beta_p r_p^z
    - O(\epsilon)^{2z} \sum_{p \in P} (1 + \gamma_p^z) r_p^z \notag \\
    &\geq (1 - 3z \epsilon) \sum_{p \in P \setminus (R\cup T)} \|p - C(p)\|^z
    - z \epsilon (20L)^z e^{-\Omega(\epsilon^2 t)} \cdot \opt(P, Q)
    - O(\epsilon)^{2z} \cdot \opt(P, Q) \notag \\
    & \geq (1 - 3z \epsilon) \sum_{p \in P \setminus (R\cup T)} \|p - C(p)\|^z
    - O(\epsilon) \cdot \opt(P, Q), \label{eqn:p_in_P_minus_R}
\end{align}
where the second last inequality follows from event \ref{item:bound_a} and \Cref{lemma:sum_of_rp}.
Finally, we combine \eqref{eqn:p_in_R},\eqref{eqn:p_in_T} and \eqref{eqn:p_in_P_minus_R} and obtain 
\begin{align*}
    \cost(G(P), G(C)) \geq (1 - O(z \epsilon)) \cost(P, C) - O(\epsilon) \cdot \opt(P, Q)
    \geq (1 - O(z \epsilon)) \cost(P, C).
\end{align*}
Rescaling $\epsilon \to \epsilon/z$ concludes the proof.
\end{proof}

\subsection{Proof of \Cref{lemma:good_events}}
\label{sec:prob_good_events}
\lemmagoodevents*

\begin{proof}
It suffices to show that each of the events happens with probability at least $0.999$, then a union bound concludes the proof.
\paragraph{Event \ref{item:bound_a}.}
We first show that $\forall \ell \in [m]$, $\E(\beta_\ell) = e^{-\Omega(\epsilon^2 t)}$.
Recall we define $\beta_\ell$ to be the minimum real, such that $\forall u, v \in N_\ell, \ \|Gu - Gv\| \geq (1 - \epsilon - \beta_\ell \epsilon) \|u - v\|$.
Then for every $h \geq 0$, we have 
\begin{align*}
    \Pr(\beta_\ell > h) &\leq 
    \Pr\Big(\exists u, v \in N_\ell, \ \|Gu - Gv\| < (1 - (h+1) \epsilon) \|u - v\|\Big) \\
    & \leq \epsilon^{-O(\ddim)} e^{-(h+1)^2 \epsilon^2 t / 8} \\
    & \leq e^{-a \cdot \epsilon^2 t (h + 1)^2}.
\end{align*}
The last inequality holds since the target dimension $t = \Omega(\epsilon^{-2} \ddim \log \epsilon^{-1})$, and thus $a$ is a constant.
Therefore, 
\begin{align*}
    \E(\beta_\ell) &= \int_0^{+ \infty} \Pr(\beta_\ell > h) \ dh 
    \leq \int_0^{+\infty} e^{-a \epsilon^2 t (h +1)^2 } \ dh 
    = \int_1^{+\infty} e^{-a\epsilon^2 t h^2} \ dh 
    \leq \int_{1}^{+\infty} h e^{-a \epsilon^2 t h^2} \ dh 
    = \frac{1}{2 a \epsilon^2 t} e^{-a \epsilon^2 t}.
\end{align*}
Hence, 
\begin{align*}
    \E\left(\sum_{p \in P} \beta_p r_p^z \right)
    = \sum_{p \in P} r_p^z \cdot \E(\beta_p) = e^{-\Omega(\epsilon^2 t)} \cdot \sum_{p \in P} r_p^z 
    = e^{-\Omega(\epsilon^2 t)} 2^z \cdot \opt(P, Q)
    \leq e^{-\Omega(\epsilon^2 t)} \cdot \opt(P, Q).
\end{align*}
By Markov's inequality, with probability at least $0.999$, $\sum_{p \in P} \beta_p r_p^z = e^{-\Omega(\epsilon^2 t)} \cdot \opt(P, Q)$.

We bound $\sum_{p \in P} \gamma_p^z r_p^z$ by the similar argument.
Recall we define $\gamma_\ell$ to be the minimum real, such that 
$\forall u \in N_\ell, v \in B(u, \epsilon^3 r_\ell), \ \|Gu - Gv\| \leq \gamma_\ell \epsilon^3 r_\ell$.
Then for every $h > 10$, we have 
\begin{align*}
    \Pr(\gamma_\ell > h) 
    &\leq \Pr\Big(\exists u \in N_\ell, v \in B(u, \epsilon^3 r_\ell), \ \|Gu - Gv\| > h \epsilon^3 r_\ell\Big) \\
    & \leq \epsilon^{-O(\ddim)} e^{-A_3 h^2 t} &&\text{by \Cref{lem:IN07_expansion}} \\
    & \leq e^{-bt h^2}, 
\end{align*}
where $b$ is a constant.
Therefore, 
\begin{align*}
    \E(\gamma_\ell^z) &= \int_0^{\infty} \Pr(\gamma_\ell^z > h) \ dh 
    \leq \int_0^{10^z} dh + \int_{10^z}^{+ \infty} e^{-bt h^{2/z}} \ dh 
    = 10^z + \int_{10}^{+\infty}z h^{z-1} e^{-bth^2} \ dh \\
    & = 10^z + z \int_{10}^{+\infty} \exp(z \ln h - b t h^2) \ dh
    \leq 10^z + \int_{10}^{+\infty} e^{-b' t h^2} \ dh \\
    & \leq 10^z + \int_{10}^{+\infty} \frac{h}{10} e^{-b' t h^2} \ dh 
    = 10^z + \frac{1}{20 b' t} e^{-100b' t} \\
    & = O(10^z).
\end{align*}
Hence, 
\begin{align*}
    \E\left(\sum_{p \in P} \gamma_p^z r_p^z \right)
    = \sum_{p \in P} r_p^z \cdot \E(\gamma_p^z) = O(10^z) \cdot \sum_{p \in P} r_p^z 
    \leq O(10^z) \cdot \opt(P, Q).
\end{align*}
By Markov's inequality, with probability at least $0.999$, $\sum_{p \in P} \gamma_p^z r_p^z = O(10^z) \cdot \opt(P, Q)$.

\paragraph{Event \ref{item:bound_c}.}
Fix $i \in [k]$, level $\ell \in I_i$ and a net point $u \in N_\ell$.
By \Cref{lem:IN07_expansion}, 
$\Pr(\exists v \in B(u, \epsilon^3 r_\ell), \ \|Gu - Gv\|> 10\epsilon^3 r_\ell) \leq e^{-\Omega(t)}$.
Noting that $|I_i| = O(\log(\alpha k))$, a union bound over all $k \cdot (\log k + \log \alpha) \cdot \epsilon^{-O(\ddim)}$ tuples $(i, \ell, u)$ completes the proof.

\paragraph{Event \ref{item:bound_d}.}
Fix a point $u \in P$.
By \Cref{lemma:preserve_sum}, with probability $1 - \epsilon^{-O(z)} k^2 e^{- \Omega(\epsilon^2 t)}$,
\begin{align*}
    \forall P'\subseteq P, \qquad \sum_{p\in P'} \|Gp-Gu\|^z 
    &\geq (1- \eps)^{3z} \sum_{p\in P'} \|p-u\|^z - \frac{\eps}{k^2} \optcont_k^z(P) \\
    &\geq (1- \eps)^{3z} \sum_{p\in P'} \|p-u\|^z - \frac{\eps}{k^2} \opt(P, Q).
\end{align*}
A union bound over all $k \cdot (\log k + \log \alpha) \cdot \epsilon^{-O(\ddim)}$ tuples $(i, \ell, u)$ completes the proof.

\paragraph{Event \ref{item:bound_e}.}
For every $i \in [k]$, by \Cref{lem:IN07_expansion}, $\Pr(\exists y \in B(c_i^*, 40L \cdot r_{\ell_i}), \ \|Gy - Gc_i^*\| > 400 L \cdot r_{\ell_i}) \leq e^{-\Omega(t)}$.
A union bound over $i \in [k]$ completes the proof.

\paragraph{Event \ref{item:bound_f}.}
For every $i \in [k]$, by \Cref{lem:IN07_variant_contraction}, $\Pr(\exists y \in (P \cup Q) \setminus B(c_i^*, 2000 L^2 \cdot r_{\ell_i}), \ \|Gy - Gc_i^*\| \leq 2000 L \cdot r_{\ell_i}) \leq e^{-\Omega(t)}$.
A union bound over $i \in [k]$ completes the proof.

\paragraph{Event \ref{item:bound_g}.}
By \Cref{lem:IN07_variant_contraction}, $\Pr(\xi_p < 9 r_{\ell_i}) \leq e^{-A_2 t}$.
Hence,
\begin{align*}
    \E(\max\{0, (9 r_{\ell_i})^z - \xi_p^z\}) \leq (9 r_{\ell_i})^z \cdot \Pr(\xi_p < 9 r_{\ell_i}) 
    \leq (9 r_{\ell_i})^z \cdot e^{-A_2 t}.
\end{align*}
By Markov's inequality and a union bound over $i \in [k]$, with probability $0.999$,
\begin{align*}
    \sum_{p \in P_{\ell_i}^i} \xi_p
    \geq \sum_{p \in P_{\ell_i}^i} (9 r_{\ell_i})^z - \sum_{p \in P_{\ell_i}^i} O(k \cdot e^{-A_2 t}) \cdot (9 r_{\ell_i})^z
    \geq 8 |P_{\ell_i}^i| \cdot r_{\ell_i}^z 
    > \alpha \opt(P, Q), 
    \qquad \forall i \in [k].
\end{align*}

\paragraph{Event \ref{item:bound_h}.}
By \Cref{lem:IN07_variant_contraction}, we have $\Pr(\eta_p < 9 r_p) \leq e^{-A_2 t}$.
Hence,
\begin{align*}
    \E(\max\{0, (9 r_p)^z - \eta_p^z \}) \leq (9 r_p)^z \cdot \Pr(\eta_p < 9 r_p) 
    \leq (9 r_p)^z \cdot e^{-A_2 t}.
\end{align*}
Summing up over $p \in S_i^*$, we have 
\begin{align*}
    \E\left(\sum_{p \in S_i^*} \max\{0, (9 r_p)^z - \eta_p^z \}\right)
    \leq O(e^{-A_2 t}) \cdot \sum_{p \in S_i^*} (9 r_p)^z.
\end{align*}
By Markov's inequality and a union bound over $i \in [k]$, with probability $0.999$,
\begin{align*}
    \sum_{p \in S_i^*} \max\{0, 9 r_p - \eta_p\}
    \leq O(k \cdot e^{-A_2 t}) \cdot \sum_{p \in S_i^*} 9^z r_p^z
    \leq e^{-\Omega(t)} \sum_{p \in S_i^*} r_p^z,
    \qquad \forall i \in [k].
\end{align*}

\end{proof}

\section{Lower bounds}
In this section, we provide our lower bounds. For simplicity, we do not try to optimize the dependence on $z$. All lower bounds are presented for $z=1$.
\subsection{Continuous, for all centers}
Denote by $0_d$ the origin of $\R^d$. For ease of presentation, we allow $P$ to be a multi-set.
\begin{theorem}\label{thm:LB_continuos_forall_centers}
    Let $n,d\in \N$, and $P=\{0_d\}^n$. Let $G\in \R^{(d-1)\times d}$ be any linear map. Then, there exists $c\in \R^{d}$ such that $\sum_{p\in P} \|Gp-Gc\|=0$ and $\sum_{p\in P} \|p-c\|=n$.
\end{theorem}
\begin{proof}
    Take a unit length vector $c\in \ker(G)$, i.e., $Gc=0$ and $\|c\|=1$. The proof follows immediately.
\end{proof}

\subsection{Discrete, for all centers}
Here, we show that in order to bound the (multiplicative) contraction for all centers, we need either dimension $\Omega(\log\log n)$, or to relax the definition of contraction (as is done in \Cref{thm:k-medoid_no_loglogn_improved}).

\begin{theorem}\label{thm:lowerbound_loglogn_medoid_forall_c}
    Let $n,d\in \N$ and $\eps\in(0,\tfrac{1}{2})$. There exists $P\subset\R^d$ of size $|P|=n$ and $\ddim(P)=\Theta(1)$, such that if $G$ is a Gaussian JL map onto dimension $t\leq a\eps^{-2}\log\log n$ for a sufficiently small constant $a>0$, then with probability at least $2/3$, there exists $c\in P$ such that $\sum_{p\in P} \|Gp-Gc\|\leq (1-\eps)\sum_{p\in P} \|p-c\|$.
\end{theorem}
To prove the theorem, we will use the following lemma, whose proof appears in \Cref{appendix:chi_square} for completeness (see \cite{ZZ20_lower_tail_bound} for a stronger bound).
\begin{restatable}{lemma}{LemmaChiSquareLowerTail}\label{lem:chi_sqaure_lower_tail}
    Let $t > 2$ be an integer, let $X_t$ be a chi-squared random variable with $t$ degrees of freedom, and  let $\eps \in (0, \tfrac{1}{2})$. We have
    \[ \Pr\left(X_t <  \frac{t}{1+\eps}\right) \ge \frac{e^{- O(\eps^2 t)}}t.\]
\end{restatable}
\begin{proof}[Proof of \Cref{thm:lowerbound_loglogn_medoid_forall_c}]

    Denote by $e_i$ the $i$-th standard basis vector. Pick $P=\{2^{-i} e_i\}_{i=0}^{\frac{1}{2}\log n} \cup \{0_d\}^{n-\frac{1}{2}\log n}$. 
    For each $i\in [0,\tfrac{1}{2}\log n]$, by \Cref{lem:chi_sqaure_lower_tail},
    \[
    \Pr(\|Ge_i\|\leq 1-\eps)\geq \tfrac{e^{-O(\eps^2 t)}}{t}\geq \tfrac{10}{\log n}.
    \]
    Therefore, the probability that for all $i\in [0,\tfrac{1}{2}\log n]$, we have $\|Ge_i\|> 1-\eps$, is at most $(1-\tfrac{10}{\log n})^{\frac{1}{2}\log n}\leq \tfrac{1}{10}$.
    Suppose this event does not happen, therefore there exists $i^*\in [0,\tfrac{1}{2}\log n]$ such that $\|Ge_{i^*}\|\leq 1-\eps$.
    Moreover, assume that $\max_{i\in [0,\tfrac{1}{2}\log n]}\|Ge_i\|\leq \log\log n$, which holds with high probability.
    Pick $c=2^{-i^*} e_{i^*}$.
    We have
    \begin{align*}
        \sum_{p\in P} \|Gp-Gc\|&\leq (n-\tfrac{1}{2}\log n)\cdot\|Gc\|+\tfrac{1}{2}\log n\cdot 2\log\log n\\
        &\leq (1-\eps)(n-\tfrac{1}{2}\log n)2^{-i^*}+\log n \cdot\log\log n \\
        &\leq (1-\tfrac{\eps}{2})(n-\tfrac{1}{2}\log n)2^{-i^*},
    \end{align*}
    where we used $2^{-i^*}\geq 2^{-\frac{1}{2}\log n}=\tfrac{1}{\sqrt{n}}$, so for sufficiently large $n$, we have $\log n\log\log n\leq \tfrac{\eps}{2\sqrt{n}}(n-\tfrac{1}{2}\log n)$.
    Moreover,
    \[
    \sum_{p\in P} \|p-c\|\geq (n-\tfrac{1}{2}\log n)2^{-i^*}.
    \]
    The proof follows by rescaling $\eps$ and combining the two inequalities.
\end{proof}

\subsection{Discrete, for all partitions and centers}
In this section, we show that dimension $\Omega(\log\log n)$ is necessary, even for the relaxed notion of contraction, for preserving all partitions and centers.

\begin{theorem}\label{thm:lowerbound_loglogn_medoid_forall_partitions_and_c}
    Let $n,d\in \N$ and $\eps\in(0,\tfrac{1}{2})$. 
    There exists $P\subset\R^d$ of size $|P|=n$ and $\ddim(P)=\Theta(1)$, such that if $G$ is a Gaussian JL map onto dimension $\tfrac{1}{a}\eps^{-2}\leq t\leq a\eps^{-2}\log\log n$ for a sufficiently small constant $a>0$, then with probability at least $2/3$, there exists $(c_1,c_2)\subset P$ and a partition $(P_1,P_2)$ of $P$ such that 
    \[
    \sum_{i\in\{1,2\}}\sum_{p\in P_i}\|Gp-Gc_i\|< \min\Big\{(1-\eps)\sum_{i\in\{1,2\}}\sum_{p\in P_i}\|p-c_i\|, 100\opt(P)\Big\}.
    \]
\end{theorem}
\begin{proof}
    Pick $P=\{(1-\eps)^{i} e_i\}_{i=0}^{\frac{1}{2}\log n} \cup \{0_d\}^{n-\frac{1}{2}\log n}$.
    Observe that $\opt = \sum_{i\in [1,\frac{1}{2}\log n]} (1-\eps)^{i}\leq\eps^{-1}$.
    For each $i\in [0,\tfrac{1}{2}\log n]$, by \Cref{lem:chi_sqaure_lower_tail},
    \[
    \Pr(\|Ge_i\|\leq 1-\eps)\geq \tfrac{e^{-O(\eps^2 t)}}{t}\geq \tfrac{10}{\log n}.
    \]
    Therefore, the probability that for all $i\in [0,\tfrac{1}{2}\log n]$, we have $\|Ge_i\|> 1-\eps$, is at most $(1-\tfrac{10}{\log n})^{\frac{1}{2}\log n}\leq \tfrac{1}{10}$.
    Suppose this event does not happen, therefore there exists $i^*\in [0,\tfrac{1}{2}\log n]$ such that $\|Ge_{i^*}\|\leq 1-\eps$.
    Moreover, assume that $\sum_{i\in [0,\frac{1}{2}\log n]} \|Ge_i\|(1-\eps)^{i}\leq (1+\eps)\sum_{i\in [0,\frac{1}{2}\log n]} (1-\eps)^{i}= (1+\eps)\opt$, which holds with high probability by \cite[Theorem A.3.1]{danos2021thesis} since $t=\Omega(\eps^{-2})$ (alternatively, one can use \Cref{lemma:1median_mmr} with $k=1$, and get that this bound holds w.h.p. if $t=\Omega(\eps^{-2}\log\tfrac{1}{\eps})$).
    Pick $c_1=0,c_2=(1-\eps)^{i^*}e_{i^*}$ and $P_2=\{0_d\}^{50\eps^{-1}\cdot (1-\eps)^{-i^*}},P_1=P\setminus P_2$. (Note that $50\eps^{-1}\cdot (1-\eps)^{-i^*}\leq 50\eps^{-1}\cdot 2^{\frac{1}{2}\log n}\leq \tfrac{n}{2}$, so this assignment is feasible.)
    We have,
    \[
    \sum_{i\in\{1,2\}}\sum_{p\in P_i}\|p-c_i\|=(1-\eps)^{i^*}\cdot 50\eps^{-1}\cdot (1-\eps)^{-i^*} + \sum_{i\neq i^*} (1-\eps)^{i}\geq 51\eps^{-1}-2;
    \]
    and
    \begin{align*}
        \sum_{i\in\{1,2\}}\sum_{p\in P_i}\|Gp-Gc_i\|&\leq (1-\eps)^{i^*}\cdot 50\eps^{-1}\cdot (1-\eps)^{-i^*}\cdot (1-\eps) + \sum_{i\neq i^*}\|Ge_i\|(1-\eps)^{i}\\
        &\leq (1-\eps)50\eps^{-1} + (1+\eps)\sum_{i\in[0,\frac{1}{2}\log n]} (1-\eps)^{i} \\
        &\leq (1-\eps)50\eps^{-1} +(1+\eps)\eps^{-1} \\
        &= 51\eps^{-1} - 49,
    \end{align*}
    which is smaller than both $\sum_{i\in\{1,2\}}\sum_{p\in P_i}\|p-c_i\|$ and $100\opt=100\eps^{-1}$, concluding the proof.
\end{proof}

\subsection{Discrete, for all centers, with candidate center set}
\begin{theorem}\label{thm:lowerbound_logn_medoid_forall_c_candidate}
    Let $n, s, d\in \N$ and $\eps\in(0,\tfrac{1}{2})$.
    There exists $P, Q \subset\R^d$ of sizes $|P|=n, |Q| = s$, and $\ddim(P \cup Q) = O(1)$,
    such that if $G$ is a Gaussian JL map onto dimension $t\leq a\eps^{-2}\log s$ for a sufficiently small constant $a>0$, then with probability at least $2/3$, there exists $c\in Q$ such that $\sum_{p\in P} \|Gp-Gc\|\leq (1-\eps)\sum_{p\in P} \|p-c\|$.
\end{theorem}

\begin{proof}
    Consider $P = \{0_d\}^n$ and $Q = \{2^i e_i\}_{i = 1}^{s}$.
    For each $i \in [s]$, by \Cref{lem:chi_sqaure_lower_tail},
    \[
    \Pr(\|Ge_i\|\leq 1-\eps)\geq \tfrac{e^{-O(\eps^2 t)}}{t}\geq \tfrac{10}{s}.
    \]
    With probability at least $1 - (1 - 10/s)^s \geq 1 - e^{-10}$, there exists $i^* \in [s]$ such that $\|G e_{i^*}\| \leq 1 - \epsilon$.
    Pick $c = 2^{i^*} e_{i^*}$. Then 
    \begin{align*}
        \cost(G(P), G c) = n \cdot 2^{i^*} \|G e_{i^*}\|
        \leq n \cdot 2^{i^*} \cdot (1 - \epsilon)
        = (1 - \epsilon) \cost(P, c).
    \end{align*}
\end{proof}

\bibliographystyle{alphaurl}
\bibliography{references}

\appendix
\section{Proof of \Cref{lemma:fixed_centers_partiton}}
\label{appendix:fixed_centers_partition}

\lemmafixedcenterspartitions*

We state the following lemma in~\cite{MakarychevMR19}, which bounds the expected distance distortion of a fixed pair of points under Gaussian JL maps.

\begin{lemma}[{\cite[Eq. (5)]{MakarychevMR19}}]
    \label{lemma:expected_distortion_MMR}
    Let $\epsilon \in (0, 1), z \geq 1$ and $G \in \R^{t \times d}$ be a Gaussian JL map.
    For $p, q \in \R^d$, 
    \begin{align*}
        \E\left(
        \max\left\{0, \frac{\|Gp - Gq\|^z}{\|p - q\|^z} - (1 + \epsilon)^z \right\}
        \right)
        \leq e^{-\Omega(\epsilon^2 t)}.
    \end{align*}
\end{lemma}

\begin{proof}[Proof of \Cref{lemma:fixed_centers_partiton}]
    For every $i \in [k]$ and $p \in S_i$, by \Cref{lemma:expected_distortion_MMR},
    \begin{align*}
        \E\left(
        \max\left\{0, \frac{\|Gp - Gc_i\|^z}{\|p - c_i\|^z} - (1 + \epsilon)^z \right\}
        \right)
        \leq e^{-\Omega(\epsilon^2 t)}.
    \end{align*}
    Therefore, 
    \begin{align*}
        \E\left(
        \max\left\{0, \|Gp - Gc_i\|^z - (1 + \epsilon)^z \|p - c_i\|^z \right\}
        \right)
        \leq e^{-\Omega(\epsilon^2 t)} \|p - c_i\|^z.
    \end{align*}
    Summing up $i \in [k]$ and $p \in S_i$, 
    \begin{align*}
        \E\left(
        \sum_{i = 1}^k \sum_{p \in S_i} \max\left\{0, \|Gp - Gc_i\|^z - (1 + \epsilon)^z \|p - c_i\|^z \right\}
        \right)
        \leq e^{-\Omega(\epsilon^2 t)} 
        \sum_{i = 1}^k \sum_{p \in S_i}\|p - c_i\|^z.
    \end{align*}
    By Markov's inequality, with probability at least $9/10$, 
    \begin{align*}
        \sum_{i = 1}^k \sum_{p \in S_i} \max\left\{0, \|Gp - Gc_i\|^z - (1 + \epsilon)^z \|p - c_i\|^z \right\}
        \leq O(e^{-\Omega(\epsilon^2 t)})
        \sum_{i = 1}^k \sum_{p \in S_i}\|p - c_i\|^z,
    \end{align*}
    which further implies
    \begin{align*}
        \cost(G(\P), G(C)) \leq ((1 + \epsilon)^z + O(e^{-\Omega(\epsilon^2 t)})) \cost(\P, C)
        \leq (1 + O(z \epsilon + e^{-\Omega(\epsilon^2 t)})) \cost(\P, C). 
    \end{align*}
    Rescaling $\epsilon$ completes the proof.
\end{proof}

\section{Proof of \Cref{lem:chi_sqaure_lower_tail}}\label{appendix:chi_square}
\LemmaChiSquareLowerTail*
\begin{proof}
Without loss of generality, we assume $t$ is even. To see why there is no loss of generality, observe that if $t$ is odd, then one can consider $t'=t+1$, which is even; we have $X_{t'}\geq X_t$ and thus $\Pr(X_{t'}<x)\leq \Pr(X_t<x)$, hence we can bound with respect to $t'$.

    The density function of $X_t$ is 
    \[ \frac{1}{2^{t/2} \Gamma(t/2)} x^{t/2-1}e^{-x/2}. \]
    Thus,
    \begin{align*}
        \Pr\left(X < \frac{t}{1-\eps}\right) &\ge \int_0^{t/(1-\eps)} \frac{1}{2^{t/2} \Gamma(t/2)} x^{t/2-1}e^{-x/2} \ dx \\
        &\ge \frac{e^{-\frac{t}{2(1-\eps)}}}{2^{t/2}\Gamma(t/2)} \int_0^{t/(1-\eps)}  x^{t/2-1} \ dx \\
        &=  \frac{e^{-\frac{t}{2(1-\eps)}}}{2^{t/2}\Gamma(t/2) t/2}  \cdot \left( \frac{t}{1-\eps} \right)^{t/2} \\
        &= \frac{t^{t/2}}{2^{t/2} \Gamma(t/2) t/2} \cdot \left( e^{-\frac{t}{2(1-\eps)}}\cdot  \frac{1}{(1-\eps)^{t/2}}  \right).
    \end{align*}

Now we analyze the first term. From \cite{proofwiki:factorialbounds}, we have that 
\[\Gamma(t/2) = (t/2-1)! \le \frac{\left(\frac{t}{2} - 1\right)^{t/2}}{e^{t/2-2}}, \]
so 
\[ \frac{t^{t/2}}{2^{t/2} \Gamma(t/2)t/2}  \ge \frac{e^{t/2-2}t^{t/2}}{2^{t/2} (t/2-1)^{t/2} t/2} \ge e^{t/2-2} \cdot \frac{t^{t/2}}{2^{t/2} (t/2)^{t/2}t/2} \ge \Omega(e^{t/2}/t).\]
Furthermore,
\[  e^{-\frac{t}{2(1-\eps)}}\cdot  \frac{1}{(1-\eps)^{t/2}} = \exp\left(- \frac{t}2 \left( \frac{1}{1-\eps} + \log(1-\eps) \right) \right) \ge \exp\left(- \frac{t}2 \left( 1 + \tfrac{\eps^2}{2} \right) \right), \]
where we used $\log(1-\eps)\leq -\eps$ and $\tfrac{1}{1-\eps}\leq 1+\eps+\tfrac{\eps^2}{2}$, which hold for $\eps\leq \tfrac{1}{2}$.

Multiplying the above two bounds and canceling the $e^{t/2}$ term, we have 
\[ \Pr\left(X < \frac{t}{1-\eps}\right) \ge \frac{e^{-O(t \eps^2)}}t, \]
as desired.
\end{proof}
\section{Tightening the known lower bounds}\label{sec:tighten_known_lowerbound}
The lower bounds of \cite{NarayananSIZ21,CW22} were stated without the dependence on $\eps$. For completeness, we add this dependence in this section. The proof is by applying \Cref{lem:chi_sqaure_lower_tail} to their hard instances.
These lower bounds hold even for just preserving the optimum cost.

\begin{theorem}
    Let $d,k\in \N$ and $\eps\in(0,\tfrac{1}{2})$. There exists $P\subset\R^d$ of size $|P|=O(k)$ and $\ddim(P)=\Theta(1)$, such that if $G$ is a Gaussian JL map onto dimension $t\leq a\eps^{-2}\log k$ for a sufficiently small constant $a>0$, then with probability at least $2/3$, $\opt(G(P))<(1-\eps)\opt(P)$.
\end{theorem}
\begin{proof}
    The proof is based on an instance by \cite{NarayananSIZ21}.
    Without loss of generality, assume $k$ is odd. Otherwise, just include a far-away point and use $k'=k-1$.
    The hard dataset instance is made of $\tfrac{k+1}{2}$ pairs $(10 i e_1, 10 i e_1 + e_i)$, where $e_i$ are the standard basis vectors. An optimal solution is to cover all but one pair using two centers, and cover one pair using one center. Thus the optimum value is just $1$ from this cost of covering a pair with one center. The instance clearly has $\ddim=O(1)$. Therefore, to preserve the cost, one has to preserve the distance for all pairs. However, the pairs are in different directions, and this is precisely the hard instance for the JL Lemma. Hence, when using target dimension $t=o(\eps^{-2}\log k)$, by \Cref{lem:chi_sqaure_lower_tail}, one pair contracts to distance below $1-\eps$ with high probability. Thus the cost reduces and the proof is concluded.
\end{proof}

\begin{theorem}
    Let $n\in \N$ and $\eps\in(0,\tfrac{1}{2})$, with $n=\Omega(\eps^{-2})$. Fix $k=2$. There exists $P\subset\R^d$ of size $|P|=n$, such that if $G$ is a Gaussian JL map onto dimension $t\leq a\eps^{-2}\log n$ for a sufficiently small constant $a>0$, then with probability at least $2/3$, $\opt(G(P))<(1-\eps)\opt(P)$.
\end{theorem}
\begin{proof}
    The proof is based on \cite{CW22} and a sketch is given in \Cref{sec:main_results}. 
    Recall that the instance is the first $n$ standard basis vectors.
    We now complete the proof sketch into a full proof.
    Let $j_1\in[n/2]$ and $j_2\in [\tfrac{n}{2}+1,n]$ be the indices minimizing $\|Ge_j\|$ in their regime.
    By \Cref{lem:chi_sqaure_lower_tail}, $\E \|Ge_{j_i}\|^2<1-\eps$ and $\E(\sqrt{\|Ge_{j_i}\|^2+1})<(1-\eps)\sqrt{2}$ for $i=1,2$.
    Next, consider $i\in [\tfrac{n}{2}]$. We have
    \begin{align*}
        \E \|Ge_{j_2}-Ge_i\| &= \E\Big(\E\big(\|Ge_{j_2}-Ge_i\| \ \big| \ Ge_{j_2}\big)\Big) && \text{law of total expectation} \\ 
        &\leq \E\left(\sqrt{\E\big(\|Ge_{j_2}-Ge_i\|^2 \ \big| \ Ge_{j_2}\big)}\right) && \text{Jensen's inequality} \\
        &= \E(\sqrt{\|Ge_{j_2}\|^2+1}) && \text{independence}
        \\
        &<  (1-\eps)\sqrt{2},
    \end{align*}
    and
    \[
    \var(\|Ge_{j_2}-Ge_i\|)\leq \E \|Ge_{j_2}-Ge_i\|^2 = \E(\|Ge_{j_2}\|^2+1)<2-\eps,
    \]
    and the same holds for $i\in [\tfrac{n}{2}+1,n]$ and $j_1$. Therefore, by Chebyshev's inequality and a union bound, with probability $2/3$, 
    \[
    \opt(G(P))\leq \sum_{i\in [\frac{n}{2}]} \|Ge_i-Ge_{j_2}\| + \sum_{i\in [\frac{n}{2}+1,n]} \|Ge_i-Ge_{j_1}\|<(1-\Omega(\eps))\sqrt{2}\cdot n.
    \]
    The proof is concluded by rescaling $\eps$.
\end{proof}
\end{document}